\documentclass[11pt]{article}
\usepackage[pdftex]{graphicx,color}
\usepackage{amsmath,amssymb,enumerate,fullpage,epsfig,multirow}
\usepackage{epsfig,epstopdf}
\usepackage{dsfont}
\usepackage{cases}
\usepackage[numbers]{natbib}

\newenvironment{proof}{\noindent\emph{Proof\ }}{\hspace*{\fill}$\Box$\medskip}

\newtheorem{theorem}{Theorem}[section]
\newtheorem{lemma}{Lemma}[section]

\newcommand{\vect}[1]{\ensuremath{\mathbf{#1}}}

\title{Smooth Inequalities and Equilibrium Inefficiency in Scheduling Games}
\author{
Johanne Cohen
\thanks{PRiSM, Universit\'e de Versailles St-Quentin-en-Yvelines, France} 
\and
Christoph D\"urr 
\thanks{LIP6, Universit\'e Pierre et Marie Curie, France}
\and
Nguyen Kim Thang
\thanks{IBISC, Universit\'e Evry Val d'Essonne, France}
}

\begin{document}

\maketitle

\begin{abstract}
  We  study coordination  mechanisms for  Scheduling Games (with
  unrelated machines).   In these
  games, each job  represents a player, who needs  to choose a machine
  for its execution,  and intends to complete earliest  possible. In an
  egalitarian  objective, the  social cost  would be  the  maximal job
  completion  time,  i.e.   the  makespan  of  the  schedule.   In  an
  utilitarian  objective,  the  social   cost  would  be  the  average
  completion time.   Instead of studying  one of those  objectives,  
  we focus on  the  more general class of $\ell_k$-norm (for
  some parameter  $k$) on job  completion times as social  cost.  This
  permits to balance overall quality of service and fairness.  In this
  setting, a coordination mechanism is a fixed policy, which specifies
  how jobs assigned to a  same machine will be scheduled.  This policy
  is known to the players and influences therefore their behavior.

  Our goal is  to design scheduling policies that  always admit a pure Nash
  equilibrium  and  guarantee  a   small  price  of  anarchy  for  the
  $\ell_{k}$-norm  social cost.  We  consider policies  with different
  amount   of  knowledge   about  jobs:   \emph{non-clairvoyant}  (not
  depending on the job processing times), \emph{strongly-local} (where
  the  schedule  of  machine  $i$  depends only  on  processing  times
  for this machine $i$ and jobs $j$ assigned  to $i$) and 
  \emph{local} (where the schedule of machine $i$  depends only on processing times 
  for all machines  $i'$ and jobs $j$ assigned  to $i$).  The analysis
  relies on the  \emph{smooth argument} together with adequate
  inequalities, called  \emph{smooth inequalities}. With  this unified
  framework, we are able to prove the following results.

  First,  we study the inefficiency in $\ell_{k}$-norm social  costs of 
  a strongly-local policy \textsf{SPT} that schedules the jobs 
  non-preemptively in order of increasing processing times
  and a non-clairvoyant policy \textsf{EQUI}
  that schedules the jobs in parallel using time-multiplexing, assigning 
  each job an equal fraction of CPU time. 
  We show  that the  price  of anarchy of
  policy \textsf{SPT}  is $O(k)$.  We  also prove a  lower bound of
  $\Omega(k/\log   k)$    for   all   deterministic,   non-preemptive,
  strongly-local and non-waiting policies (non-waiting policies
  produce schedules without idle times).  These results ensure that
  \textsf{SPT}  is  \emph{close  to   optimal}  with  respect  to  the
  \emph{class}  of $\ell_{k}$-norm social  costs.  Moreover,  we prove
  that the non-clairvoyant policy  \textsf{EQUI}  has  price  of
  anarchy $O(2^{k})$.

  Second, we consider  the makespan ($\ell_{\infty}$-norm) social cost
  by making connection within the $\ell_{k}$-norm functions.  We revisit
  some  local policies and  provide simpler,  unified proofs  from the
  framework's point  of view. With  the highlight of the  approach, we
  derive a local policy  \textsf{Balance}.  This policy guarantees a price of
  anarchy  of $O(\log  m)$, which  makes it  the currently best
   known policy among the anonymous local policies that always
  admit a pure Nash equilibrium.  
\end{abstract}

%
%
\section{Introduction}

With the  development of the Internet,  large-scale systems consisting
of   autonomous  decision-makers  (players)   become  more   and  more
important.  The rational behavior of players who compete for the usage
of shared  resources generally leads  to an unstable  and inefficient
outcome.  This creates a  need for \emph{resource usage policies} that
guarantee stable and near-optimal outcomes.

From a game theoretical point of view, stable outcomes are captured by
the concept of
\emph{Nash equilibria}.  Formally, in a game with $n$ players, each player $j$
chooses a strategy $x_{j}$ from a  set $S_{j}$ and this induces a cost
$c_{j}(\vect{x})$  for  player  $j$  depending all  chosen  strategies
$\vect{x}$.  A strategy profile $\vect{x} = (x_{1},\ldots,x_{n})$ is a
\emph{pure Nash equilibrium} if no  player can decrease its cost by an
unilateral    deviation,    i.e.,    $c_{j}(x'_{j},    x_{-j})    \geq
c_{j}(\vect{x})$ for  every player $j$  and $x'_{j} \in  S_{j}$, where
$x_{-j}$ denotes the strategies selected by players different from $j$.

The \emph{better-response dynamic} is the process of repeatedly 
choosing an arbitrary player that can improve its cost and let it 
take a better strategy while other player strategies remain unchanged.  
It is desirable that in a game the better-response dynamic converges 
to a Nash equilibrium as it is a natural way that
selfish behavior leads the game to a stable outcome.
A \emph{potential game} is a game in which for 
any instance, the better-response dynamic always 
converges~\cite{MondererShapley:Potential-Games}.

A  standard measure  of  inefficiency is  the  \emph{price of  anarchy
  (PoA)}.  Given  a game  with an objective  function and a  notion of
equilibrium  (e.g pure  Nash  equilibrium),  the PoA  of  the game  is
defined as  the ratio between the  largest cost of  an equilibrium and
the  cost  of  an optimal  profile,  which  is not necessarily an
equilibrium.   The  PoA  captures   the  worst-case  paradigm  and  it
guarantees the efficiency of every equilibrium.

The  social cost  of a  game is  an objective  function  measuring the
quality of  strategy profiles.  In  the literature there are  two main
extensively-studied  objective  functions:  (i) the  \emph{utilitarian
  social  cost}  is  the   total  individual  costs;  while  (ii)  the
\emph{egalitarian social  cost} is  the maximum individual  cost.  The
two  objective functions  are included  in a  general class  of social
costs:  the class  of $\ell_{k}$  norms of  the individual  costs, with
utilitarian  and the  egalitarian  social costs  corresponding to  the
cases $k  = 1$  and $k =  \infty$, respectively.   There is a  need to
design policies that  guarantee the efficiency (e.g the  PoA) of games
under some  specific objective function. Moreover, it would be interesting
to come  up with a  policy, that would be  efficient for every
social costs from this class. 

\subsection{Coordination        Mechanisms        in        Scheduling
  Games} \label{sec:coor-def}

In a  scheduling game, there are  $n$ jobs and $m$  unrelated machines. Each job
needs  to  be scheduled  on  exactly  one  machine.  We  consider  the
unrelated  parallel  machine  model,   where  each  machine  could  be
specialized for a different type of jobs. In this general setting, the
processing  time of job  $j$ on  machine $i$  is some  given arbitrary
value $p_{ij}  > 0$.  A  strategy profile $\vect{x} =  (x_{1}, \ldots,
x_{n})$ is an assignment of jobs to machines, where $x_{j}$ denotes the
machine (strategy) of job $j$ in the profile.  The \emph{cost} $c_{j}$
of  a job  $j$  is its  completion  time and  every job  strategically
chooses a machine to minimize the  cost.  In the game, we consider the
social  cost  as the  $\ell_{k}$-norm  of  the  individual costs.   The
\emph{social   cost}   of  profile   $\vect{x}$   is  $C(\vect{x})   =
\left(\sum_{j}  c_{j}^{k}\right)^{1/k}$.  

 
The traditional  $\ell_{1}, \ell_{\infty}$-norms represent  the total completion time  
and the makespan, respectively. Both objectives are natural.
Minimizing  the   total  completion  time  guarantees a quality of service
while minimizing the makespan ensures the fairness of schedule.
Unfortunately,  in   practice  schedules  which   optimize  the  total
completion  time  are not  implemented  due  to  a lack  of  fairness and
vice versa. Implementing a fair schedule is  one of the highest priorities in most
systems~\cite{Tanenbaum:Modern-Operating}.   A  popular and  practical
method  to enforce  the  fairness of  a  schedule is  to optimize  the
$\ell_{k}$-norm of completion times  for some fixed $k$, which usually
is chosen  as $k$ small constant.   By  optimizing the  $\ell_{k}$-norm of
completion time, one balances overall quality of service and fairness, 
which is generally desirable.
So the system takes into account a trade-off between quality of service and 
fairness by optimizing the $\ell_{k}$-norm  of completion 
time~\cite{SilberschatzGalvin:Operating-System,Tanenbaum:Modern-Operating}.

A   \emph{coordination  mechanism}  is   a  set   of  \emph{scheduling
  policies}, one for each machine, that determine how to schedule the
jobs assigned to a machine.
The idea is to connect the individual cost
to the social  cost, in such a way that the  selfishness of the agents
will lead to equilibria with small social cost.  We distinguish
between   \emph{local,   strongly-local}  and   \emph{non-clairvoyant}
policies.  These  policies are classified  in the decreasing  order of
the amount  of information  that ones could  use for  their decisions.
Formally, let $\vect{x} = (x_{1}, \ldots, x_{n})$ be a profile.
\begin{itemize}
\item A  policy is \emph{local} if  the scheduling of  jobs on machine
  $i$ depends only on the processing times of jobs assigned to the machine,
  i.e.,  $\{p_{i'j}: x_{j} = i, 1 \leq i' \leq
  m\}$.
\item A policy  is \emph{strongly-local} if the policy  of machine $i$
  depends only  on the processing times  for this machine  $i$ for all
  jobs assigned to $i$, i.e., $\{p_{ij}: x_{j} = i\}$.
\item A policy is \emph{non-clairvoyant}  if the scheduling of jobs on
  machine $i$ does  not depend on any processing  time.  Such a policy
  actually assigns  CPU time slots  to jobs, and is  informed whenever
  some  job  completes.   The  resulting schedule  then  reflects  the
  processing times of its jobs,  even though the policy is not aware
  of them prior to the job completions.
\end{itemize}
In addition, a policy is \emph{anonymous} if it does not use any
\emph{global} ordering of jobs or any \emph{global} job identities. 
Note that for any deterministic policy, \emph{local} job identities are necessary 
as a machine may need such information in order to break ties
(a job may have different identities on different machines).
Moreover, we call a policy \emph{non-waiting} if  the schedule contains no
idle time between job executions.


Instead  of  specifying the  actual  schedule,  we  rather describe  a
scheduling  policy  as a  function,  mapping  every  job $j$  to  some
completion time $c_j(\vect x)$.  Such a policy is said \emph{feasible}
if for  any profile $\vect{x}$, there  exists a schedule  where job $j$
completes at time  $c_j(\vect x)$. Formally, for any  job $j$, we must
have $c_{j}(\vect{x})  \geq \sum_{j'} p_{ij'}$  where the sum  is take
over  all  jobs  $j'$  with  $x_j=x_{j'}$,  and  $c_{j'}(\vect{x})  \leq
c_{j}(\vect{x})$.  Certainly, any  designed deterministic policy needs
to be feasible.

\subsection{Overview \& Contributions}
Recently,   \citet{Roughgarden:Intrinsic-robustness}   developed  the
\emph{smoothness argument}, a unifying method to show upper bounds the
PoA for  utilitarian games.  This canonical method  is elegant  in its
simplicity and  its power.  Here we  give a brief  description of this
argument.

A cost-minimization game with the total cost objective 
$C(\vect{x}) = \sum_{j} c_{j}(\vect{x})$ is $(\lambda,\mu)$-\emph{smooth} 
if for every profile $\vect{x}$ and $\vect{x}^{*}$,
\begin{equation*}	\label{eq:smooth}
\sum_{j} c_{j}(x^{*}_{j},x_{-j}) \leq \mu \sum_{j} c_{j}(\vect{x}) + \lambda \sum_{j} c_{j}(\vect{x}^{*})
\end{equation*}
The   smooth  argument~\cite{Roughgarden:Intrinsic-robustness}  states
that the  robust price of anarchy  (including the PoA  of pure, mixed,
correlated equilibria, etc) of a cost-minimization game is bounded by
$$
\inf \left\{ \frac{\lambda}{1 - \mu}: \lambda \geq 0, \mu < 1, \text{ the game is } (\lambda,\mu) \text{-smooth} \right\}.
$$

We will make use of this argument  to settle the equilibrium inefficiency
in scheduling games.  We will prove the robust PoA
by applying  the smooth argument  to the game with  $C^{k}(\vect{x}) =
\sum_{j}    c_{j}^{k}(\vect{x})$    where    $C(\vect{x})$   is    the
$\ell_{k}$-norm social  cost of Scheduling Games.  The main difficulty
in applying the smooth argument to Scheduling Games has arisen from 
the fact that jobs on the same machine have different costs, 
which is in contrast to Congestion Games where players incurs
the same cost at the same resource.  The key technique in
this   paper  is   a  system   of  inequalities,   called  \emph{smooth inequalities}, 
that are useful to prove the smoothness of the game. With the inequalities,
we are able to analyze systematically and in unified manner 
the PoA of the game under different policies.

Our contributions are the following:

\begin{enumerate}
\item We  study the  equilibrium inefficiency for  the $\ell_{k}$-norm
  objective  function. We consider a non-clairvoyant policy \textsf{EQUI} 
  that schedules the jobs in parallel using time-multiplexing, assigning 
  each job an equal fraction of CPU time; and a strongly-local policy \textsf{SPT}
  that schedules the jobs non-preemptively in order of increasing processing times
  (with a deterministic tie-breaking rule for each machine)\footnote{Formal 
  definitions of \textsf{EQUI} and \textsf{SPT} are given in Section~\ref{sec:lk}}. 
  We prove  that the  PoA of  the game  under the
  non-clairvoyant  policy  \textsf{EQUI}   is  at  most  $O(2^{k})$.
  Besides, the PoA of  the game under the deterministic strongly-local
  policy \textsf{SPT} is at  most $O(k)$.  Moreover, any deterministic
  non-preemptive, non-waiting  and strongly-local policy has  a PoA at
  least  $\Omega(k/\log k)$, which  is close  to the  PoA of  the game
  under  the  \textsf{SPT}  policy.   Hence, for  any  $\ell_{k}$-norm
  social  cost, \textsf{SPT}  is  close to optimal among
  deterministic non-preemptive, non-waiting, strongly-local  policy.  
  (The  cases  $k=1$  and $k  = \infty$  are  confirmed in  \cite{ColeCorreaGkatzelis:Inner-Product}
  and   \cite{AzarJainMirrokni:Almost-Optimal-Coordination,ImmorlicaLiMirrokni:Coordination-mechanisms},         respectively.)
	%
  If one considers theoretical evidence to classify algorithms for practical use 
  then \textsf{SPT} is a good candidate due to its simplicity and theoretically
  guaranteed performance on any combination of the quality and the fairness of 
  schedules.   
  
\item We study the equilibrium inefficiency for the makespan objective
  function (e.g., $\ell_{\infty}$-norm) for local  policies by making connection
  between $\ell_{k}$-norm functions. First, we
  revisit     policies      \textsf{BCOORD,     CCOORD}     introduced
  in~\cite{Caragiannis:Efficient-coordination-mechanisms}. We  give
  unified and simpler  proofs based on the smooth  arguments.  
  With the highlight of  this approach, we derive a new
  policy \textsf{Balance} (definition is given is Section~\ref{sec:makespan}).   
  The game under that  policy always admits
  Nash  equilibrium  and  induces  the  PoA of  $O(\log  m)$  ---  the
  currently best performance among anonymous local policies 
  that always possess pure Nash equilibria.  

\begin{figure}[ht]
  \centering
  \begin{tabular}{|l|l|l|l|}
    \hline
    \emph{Objective}  & \emph{Policy}  & \emph{Pure  Nash  equilibria} &  \emph{PoA}  \\
    \hline
  $\ell_k$-norm &
    \textsf{EQUI} (non-clairvoyant) & potential game & $O(2^k)$ \\ \cline{2-4}
    & \textsf{SPT} (strongly-local) & potential game & $O(k)$ \\ \cline{1-4}
   $\ell_\infty$-norm & 
     \textsf{Balance} (local) & potential game & $O(\log m)$ \\  \cline{1-4}
  \end{tabular}
  \caption{Main contributions of the paper.}
  \label{fig:contrib}
\end{figure}
	
\end{enumerate}

Our results naturally extend to the case when jobs have weights and the
objective   is  the  weighted   $\ell_k$-norm  of   completion  times,
i.e., $(\sum_j (w_j c_j(\vect x))^k)^{1/k}$.

\subsection{Related results}

The      smooth     argument      has      been     formalized      in
\cite{Roughgarden:Intrinsic-robustness}.   It   has  been  used  to
establish     tight     PoA     of     congestion
games~\cite{Rosenthal:A-Class-of-Games-Possessing}, 
a     fundamental    class     of     games.  
The argument is also applied to prove bounds on the PoA  of  weighted
congestion     games  \cite{BhawalkarGairingRoughgarden:Weighted-Congestion}.  
Subsequently, \citet{RoughgardenSchoppmann:Local-Smoothness}   have   extended   the
argument  to  prove tight  bounds  on  the  PoA of  atomic  splittable
congestion games for a large class of latencies.

Coordination   mechanisms   for   scheduling  games   was   introduced
in~\cite{ChristodoulouKoutsoupiasNanavati:Coordination-Mechanisms}
where  the makespan  ($\ell_{\infty}$-norm) objective  was considered.
For              the             non-clairvoyant             policies,
\citet{CohenDurrThang:Non-clairvoyant-Scheduling}   studied  the  game
under  various  policies and  derived  the  policy \textsf{EQUI}  that
always  admits  a  Nash  equilibrium  and has  an  optimal  PoA.   For
strongly-local                                                policies,
\citet{ImmorlicaLiMirrokni:Coordination-mechanisms}  gave a  survey on
the  existence   and  inefficiency  of  different   policies  such  as
\textsf{SPT}, \textsf{LPT}, \textsf{RANDOM}.  Some tight bounds on the
PoA       under       different       policies       were       given.
\citet{AzarJainMirrokni:Almost-Optimal-Coordination}   initiated   the
study on  local policies. They  designed a non-preemptive  policy with
PoA  of $O(\log  m)$ and  a preemptive  policy that  always  admits an
equilibrium and  guarantees a  PoA of $O(\log^{2}  m)$.  Subsequently,
\citet{Caragiannis:Efficient-coordination-mechanisms}   derived a non-anonymous
local  policy \textsf{ACOORD} and anonymous local policies \textsf{BCOORD}  
and \textsf{CCOORD} with PoA of  $O(\log m)$, $O(\log m/\log\log m)$  
and $O(\log^{2} m)$, respectively where the first and the last ones always admit a Nash equilibrium.
\citet{FleischerSvitkina:Preference-constrained-oriented}   showed   a
lower bound of $\Omega(\log  m)$ for all deterministic non-preemptive,
non-waiting local policies.

Recently,  \citet{ColeCorreaGkatzelis:Inner-Product} studied  the game
with total completion time ($\ell_{1}$-norm)  objective.  They considered strongly-local
policies  with  weighted jobs,  and  derived  a non-preemptive  policy
inspired by the Smith's rule  which has PoA $=4$. This bound is tight for 
deterministic non-preemptive non-waiting strongly-local policies. Moreover, some
preemptive  policies   are  also  designed   with  better  performance
guarantee.

\subsection{Organization} 
In Section~\ref{sec:smooth}, we state some smooth inequalities that will be
used   in    settling   the    PoA   for   different    policies.   In
Section~\ref{sec:lk}, we study the scheduling game with the
$\ell_{k}$-norm social cost.  We define the policies
\textsf{SPT} and \textsf{EQUI}, and prove their inefficiency.  We also
provide an lower bound on the PoA for any deterministic non-preemptive
non-waiting strongly-local policy.  In Section~\ref{sec:makespan}, we consider the
makespan  ($\ell_{\infty}$-norm) social cost  for local  policies.  
We   revisit   the   policies   \textsf{BCOORD}   and
\textsf{CCOORD}~\cite{Caragiannis:Efficient-coordination-mechanisms};
define and analyze the performance of policy \textsf{Balance}.
The proofs of all lemmas and theorems are either presented in the main corp
of the paper or given in the appendix.

\section{Smooth Inequalities}		\label{sec:smooth}

In this section, we show various inequalities that are useful for the analysis. 
Specifically, Lemma~\ref{lem:smooth-SPT} 
and Lemma~\ref{lem:smooth-EQUI} are applied directly to prove the PoA of
policies \textsf{SPT} and \textsf{EQUI}, respectively in Section~\ref{sec:lk}. 
The other lemmas are used to prove Lemma~\ref{lem:smooth-SPT},
Lemma~\ref{lem:smooth-EQUI} and theorems in Section~\ref{sec:makespan}.
 
First, we give a definition of \emph{Lambert $W$ function} that we use 
throughout the paper.  For each $y \in \mathbb{R}^{+}$, $W(y)$ is 
defined to be  solution of  the equation $xe^{x}  = y$.  Note that, 
$xe^{x}$  is increasing with respect to $x$, hence $W(\cdot)$ is increasing. 

\begin{lemma}
  \label{lem:smooth-simple}
Let $k$ be a positive integer.
Let $0< a(k) \leq 1$ be a function on $k$. Then, for any  $x, y > 0$, it holds that
$$
y(x+y)^{k} \leq \frac{k}{k+1}a(k) x^{k+1} + b(k) y^{k+1} 
$$
where $\alpha$ is some constant and
\begin{subnumcases}
 {\label{b} b(k) =} 
	 \Theta\left(\alpha^{k} \cdot \left(\frac{k}{\log ka(k)} \right)^{k-1}\right) & \text{if} 
		$\lim_{k\rightarrow\infty}(k-1)a(k) = \infty$,	\label{b1} \\
	 \Theta\left(\alpha^{k} \cdot k^{k-1}\right)	 & \text{if} $(k-1)a(k)$ \text{are bounded} $\forall k$, \label{b2}\\
	 \Theta\left(\alpha^{k} \cdot \frac{1}{ka(k)^{k}}\right) & \text{if}
		$\lim_{k\rightarrow\infty} (k-1)a(k) = 0$.	\label{b3} 
\end{subnumcases}
\end{lemma}
\begin{proof}
Let $f(z) := \frac{k}{k+1}a(k) z^{k+1} - (1 + z)^{k} + b(k)$. 
To show the claim, we equivalently prove that $f(z) \geq 0$ for all $z > 0$. 

We have $f'(z) = ka(k)z^{k} - k(1+z)^{k-1}$. Let $z_{0}$ be the unique
positive root of $f'(z) = 0$.  Function $f$ is decreasing in $(0,z_{0})$ and 
increasing in $(z_{0}, +\infty)$, so $f(z) \geq f(z_{0})$ for all $z > 0$. 
Hence, by choosing 
$$
b(k) = \Big | \frac{k}{k+1}a(k) z_{0}^{k+1} - (1 + z_{0})^{k} \Big |
= (1+z_{0})^{k-1}\Big(1+\frac{z_{0}}{k+1}\Big)
$$
it follows that $f(z) \geq 0 ~\forall z>0$.

We study the positive root $z_{0}$ of equation 
\begin{equation}	\label{eq:smooth-simple-derivative}
a(k) z^{k} - (1+z)^{k-1} = 0
\end{equation}
Note that $f'(1) = a(k) - 2^{k-1} < 0$ since $0< a(k) \leq 1$. 
Thus, $z_{0} > 1$.  For the sake of simplicity, we
define the function $g(k)$ such that $z_{0} = \frac{k-1}{g(k)}$ where
$0 < g(k) < k-1$.  Equation~(\ref{eq:smooth-simple-derivative}) is
equivalent to
\begin{equation*} \label{eq:majoration}
\left(1 + \frac{g(k)}{k-1}\right)^{k-1}g(k) = (k-1)a(k)
\end{equation*}
Note that $e^{w/2} < 1 + w < e ^{w}$ for $w \in (0,1)$.  For $w := \frac{g(k)}{k-1}$,
we obtain the following upper and lower bounds for the term  $(k-1)a(k)$:
 
\begin{equation} \label{eq:LambertW}
e^{g(k)/2}g(k) < (k-1)a(k) < e^{g(k)}g(k)
\end{equation}

By definition of the Lambert $W$ function and
Equation~(\ref{eq:LambertW}), we get that
\begin{equation} \label{eq:encadrementG}
W\left((k-1)a(k)\right) < g(k) < 2W\left(\frac{(k-1)a(k)}{2}\right)
\end{equation}


First, consider  the case where  $\lim_{k\rightarrow\infty}(k-1)a(k) =
\infty$.  The asymptotic sequence for $W(x)$ as $x \to +\infty$ is the
following:    $W(x)=\ln    x-\ln   \ln    x+\frac{\ln\ln    x}{\ln
  x}+O\left(\left(\frac{\ln\ln x}{\ln x}\right)^2\right)$.  So, for large enough $k$, $W((k-1)a(k)) = \Theta(\log((k-1)a(k)))$.  Since
$z_{0} =  \frac{k-1}{g(k)}$, from Equation~(\ref{eq:encadrementG}), we
get $z_{0}  = \Theta\left( \frac{k}{\log (ka(k))}  \right)$.  This
guarantees $f(z_{0}) \geq 0$ and $b(k) =
\Theta\left(\alpha^{k}\cdot\left(\frac{k}{\log                   ka(k)}
  \right)^{k-1}\right)$ for some constant $\alpha$.

Second,  consider  the  case  where  $(k-1)a(k)$ is  bounded  by  some
constants.  So by~(\ref{eq:encadrementG}), we have $g(k) = \Theta(1)$.
Therefore $z_{0} = \Theta(k)$.   Which again implies $f(z_{0}) \geq 0$
and  $b(k)  =  \Theta\left(\alpha^{k}\cdot  k^{k-1}\right)$  for  some
constant $\alpha$.

Third, we consider the case where $\lim_{k\rightarrow\infty}(k-1)a(k) = 0$. 
We focus on the Taylor series $W_0$ of $W$ around 0.
It can be found using the Lagrange inversion and is given by
$$
W_{0}(x) = \sum_{i=1}^{\infty}\frac{(-i)^{i-1}}{i!}x^{i} = x - x^{2} + O(1)x^{3}.
$$
Thus, for $k$ large  enough $g(k) = \Theta((k-1)a(k))$.  Hence, $z_{0}
=  \Theta(1/a(k))$.  Once  again this  implies $f(z_{0})  \geq  0$ and
$b(k)   =  \Theta\left(\alpha^{k}\cdot\frac{1}{ka(k)^{k}}\right)$  for
some constant $\alpha$.
\end{proof}


Note that the case (\ref{b1}) of Lemma~\ref{lem:smooth-simple} could be used to 
settle the tight bound on the PoA of Congestion Games
in which delay functions are polynomials with positive coefficients. 
\cite{SuriTothZhou:Selfish-Load} proved this case for $a(k) = 1$
and $b(k) = \Theta(\frac{1}{k}(k/\log k)^{k})$ in order to
upper bound of the PoA in Selfish Load Balancing Games. 

\begin{lemma}		\label{lem:simple}
It holds that $(k+1)z \geq 1 - (1 - z)^{k+1}$ for all $0 \leq z \leq 1$ and for all $k \geq 0$.
\end{lemma}
\begin{proof}
Consider $f(z) = (k+1)z - 1 + (1 - z)^{k+1}$ for $0 \leq z \leq 1$. We have 
$f'(z) = (k+1) - (k+1)(1-z)^{k} \geq 0 ~\forall 0 \leq z \leq 1$. So $f(z) \geq f(0) = 0$. Thus,
$(k+1)z \geq 1 - (1 - z)^{k+1}$ for all $0 \leq z \leq 1$.
\end{proof}

In the following, we prove inequalities to bound
the PoA of the scheduling game. Remark that until the end of the section,
we use $i,j$ as the indices. The following is the main lemma to show 
the upper bound $O(k)$ of the PoA under policy \textsf{SPT} in the next section.

 %
%

\begin{lemma}  \label{lem:smooth-SPT} For  any  non-negative sequences
  $(n_{i})_{i=1}^{P}$,  $(m_{i})_{i=1}^{P}$,   and  for  any  positive
  increasing   sequence   $(q_{i})_{i=1}^{P}$,   define  $A_{i,j}   :=
  n_{1}q_{1}  + \ldots  + n_{i-1}q_{i-1}  + j\cdot  q_{i}$  for $1\leq
  i\leq P, 1 \leq j \leq  n_{i}$ and $B_{i,j} := m_{1}q_{1} + \ldots +
  m_{i-1}q_{i-1} + j\cdot  q_{i}$ for $1 \leq i \leq P,  1 \leq j \leq
  m_{i}$.  Then, it holds that
$$
\sum_{i=1}^{P} \sum_{j=1}^{m_{i}} (A_{i,n_{i}} + j\cdot q_{i})^{k} 
\leq \mu_{k} \sum_{i=1}^{P} \sum_{j=1}^{n_{i}} A_{i,j}^{k} 
	+ \lambda_{k} \sum_{i=1}^{P} \sum_{j=1}^{m_{i}} B_{i,j}^{k},
$$
where    $\mu_{k}    =    \frac{k+1}{k+2}$    and    $\lambda_{k}    =
\Theta(\alpha^{k}(k+1)^{k})$  for some  constant  $\alpha$. 
\end{lemma}
\begin{proof}
Denote $r_{i} = 1/q_{i}$ for $1 \leq i \leq P$. The inequality is equivalent to 
$$
\sum_{i=1}^{P} r_{i} \sum_{j=1}^{m_{i}} q_{i}(A_{i,n_{i}} + j\cdot q_{i})^{k} 
\leq \mu_{k} \sum_{i=1}^{P} r_{i}  \sum_{j=1}^{n_{i}} q_{i} A^{k}_{i,j} 
  + \lambda_{k} \sum_{i=1}^{P} r_{i}  \sum_{j=1}^{m_{i}} q_{i} B^{k}_{i,j}
$$
For convenience set $r_{P+1} = 0$. This inequality could be written as 
\begin{align*}
\sum_{i=1}^{P} & (r_{i} - r_{i+1}) \left[\sum_{t=1}^{i} \sum_{j=1}^{m_{t}} q_{t}(A_{t,n_{t}} + j \cdot q_{t})^{k} \right] \\
&\leq \sum_{i=1}^{P} (r_{i} - r_{i+1}) 
	\left[\mu_{k} \sum_{t=1}^{i} \sum_{j=1}^{n_{t}} q_{t}A_{t,j}^{k} 
				+ \lambda_{k} \sum_{t=1}^{i} \sum_{j=1}^{m_{t}} q_{t}A_{t,j}^{k} \right]
\end{align*}

As $(r_{i})_{i=1}^{P}$ is decreasing a sequence (so $r_{i} - r_{i+1} \geq 0 ~\forall1 \leq i \leq P$), 
it is sufficient to prove that for all $1 \leq i \leq P$, 
\begin{equation}	\label{ineq:EQUI-bA}
\sum_{t=1}^{i} \sum_{j=1}^{m_{t}} q_{t}(A_{t,n_{t}} + j \cdot q_{t})^{k}
\leq \mu_{k} \sum_{t=1}^{i} \sum_{j=1}^{n_{t}} q_{t}A_{t,j}^{k} 
				+ \lambda_{k} \sum_{t=1}^{i} \sum_{j=1}^{m_{t}} q_{t}A_{t,j}^{k}.
\end{equation}
It remains to show Inequality~(\ref{ineq:EQUI-bA}).

For convenience set $A_{0,j} = 0$ for any $j$.
By Lemma~\ref{lem:simple}, we have
\begin{align*}
(k+1)q_{t}A^{k}_{t,j} &\geq A^{k+1}_{t,j} - (A_{t,j} - q_{t})^{k+1} = A^{k+1}_{t,j} - A^{k+1}_{t,j-1}
	\quad \forall 1 \leq t \leq P, 2 \leq j \leq n_{t} \\
(k+1)q_{t}A^{k}_{t,1} &\geq A^{k+1}_{t,1} - (A_{t,1} - q_{t})^{k+1} = A^{k+1}_{t,1} - A^{k+1}_{t-1,n_{t-1}}
	\quad \forall 1 \leq t \leq P.
\end{align*}
Therefore,
$$
(k+1)\sum_{t=1}^{i} \sum_{j=1}^{n_{t}} q_{t}A_{t,j}^{k} 
\geq \sum_{t=1}^{i} \left[\sum_{j=2}^{n_{t}} \left(A_{t,j}^{k+1} - A_{t,j-1}^{k+1}  \right)
		+ A^{k+1}_{t,1} - A^{k+1}_{t-1,n_{t-1}}\right] 
= A^{k+1}_{i,n_{i}},
$$
since the sums telescope. Similarly,
$(k+1)\sum_{t=1}^{i} \sum_{j=1}^{m_{t}} q_{t}B_{t,j}^{k} 
\geq B^{k+1}_{i,m_{i}}$.
Thus, to prove Inequality (\ref{ineq:EQUI-bA}), it is sufficient to prove that  for all 
$1 \leq i \leq P$,
$$	
\sum_{t=1}^{i} \sum_{j=1}^{m_{t}} q_{t}(A_{t,n_{t}} + j \cdot q_{t})^{k}
\leq \left( \frac{\mu_{k}}{k+1} A^{k+1}_{i,n_{i}} + \frac{\lambda_{k}}{k+1} B^{k+1}_{i,m_{i}} \right)
$$
Besides,  
\[
\sum_{t=1}^{i} \sum_{j=1}^{m_{t}} q_{t}(A_{t,n_{t}} + j \cdot q_{t})^{k} 
\leq \sum_{t=1}^{i} m_{t} q_{t}(A_{t,n_{t}} +  B_{i,m_{i}})^{k} 
\leq  B_{i,m_{i}}(A_{i,n_{i}} + B_{i,m_{i}})^{k}
\]
Hence, we only need to argue that 
\begin{equation}	\label{ineq:smooth-SPT}
B_{i,m_{i}}(A_{i,n_{i}} + B_{i,m_{i}})^{k} 
\leq \left( \frac{\mu_{k}}{k+1} A^{k+1}_{i,n_{i}} + \frac{\lambda_{k}}{k+1} B^{k+1}_{i,m_{i}} \right)
\end{equation}

Choose $\mu_{k} = \frac{k+1}{k+2}$ and apply case~(\ref{b2}) of Lemma~\ref{lem:smooth-simple} 
(now $a(k) = \frac{(k+1)}{k(k+2)}$ and $(k-1)a(k)$ is bounded by a constant), we deduce that: 
for $\lambda_{k} = \Theta(\alpha^{k}(k+1)^{k})$ where $\alpha$ is a constant, 
Inequality~(\ref{ineq:smooth-SPT}) holds. 
\end{proof}


The following is the main lemma to settle the bound $O(2^{k})$ of the PoA under 
policy \textsf{EQUI} (the proof is in the appendix).

\begin{lemma}		\label{lem:smooth-EQUI}
For any non-negative sequences $(n_{i})_{i=1}^{P}$, $(m_{i})_{i=1}^{P}$, and for any positive increasing
sequence $(q_{i})_{i=1}^{P}$, define $A_{i} = n_{1}q_{1} + \ldots + n_{i-1}q_{i-1} + (n_{i} + \ldots + n_{P})q_{i}$
and $B_{i} = m_{1}q_{1} + \ldots + m_{i-1}q_{i-1} + (m_{i} + \ldots + m_{P})q_{i}$ for $1 \leq i \leq P$.
Then, it holds that
$$
\sum_{i=1}^{P}m_{i}(A_{i} + m_{i}q_{i})^{k} 
\leq \mu_{k} \sum_{i=1}^{P} n_{i}A_{i}^{k} + \lambda_{k} \sum_{i=1}^{P} m_{i} B_{i}^{k}
$$
where
$\mu_{k} = \frac{k+1}{k+2}$, $\lambda_{k} = \Theta(\alpha^{k}2^{(k+1)^{2}})$
for some constant $\alpha$.  
\end{lemma}


\section{$\ell_{k}$-norms of Completion Times} 	\label{sec:lk}

We  study coordination mechanisms under two policies:
(1)  the strongly-local  policy  \textsf{SPT} that  schedules jobs  in
increasing  order of  processing  times; and  (2) the  non-clairvoyant
policy  \textsf{EQUI}  that  schedules  the  jobs  in  parallel  using
time-multiplexing, assigning each job an equal fraction of CPU time.

\paragraph{Policy \textsf{SPT}}\label{def:SPT}
Let $\vect{x}$ be  a strategy profile. Let $\prec_{i}$  be an order of
jobs on machine $i$,  where $j'  \prec_{i} j$  iff  $p_{ij'} <
p_{ij}$ or $p_{ij'} = p_{ij}$ and  $j$ is priority over $j'$ (machine $i$
chooses a local preference over jobs based on their local identities to break ties). 
The cost of job $j$ under the \textsf{SPT}~\cite{ImmorlicaLiMirrokni:Coordination-mechanisms}  
policy is
\begin{align*}
c_{j}(\vect{x}) = \sum_{\substack{j': ~ x_{j'} = i\\ j' \preceq j}} p_{ij'}.
\end{align*}

\paragraph{Policy \textsf{EQUI}}\label{def:EQUI}
Let $\vect{x}$ be a strategy profile. The cost of job $j$ under the \textsf{EQUI} 
policy~\cite{CohenDurrThang:Non-clairvoyant-Scheduling} is
\begin{align*}
c_{j}(\vect{x}) =  \sum_{\substack{j': ~ x_{j'} = i \\ p_{ij'} < p_{ij} }} p_{ij'}
			+ \sum_{\substack{j': ~ x_{j'} = i \\ p_{ij'} \geq p_{ij} }} p_{ij}
		= \sum_{j': ~ x_{j'} = i} \min\{p_{ij'}, p_{ij}\}
\end{align*}

Note that the two  policies  \textsf{SPT} and \textsf{EQUI} are feasible. 
Since all $p_{ij}$ could be written  as a multiple of $\epsilon$ (a small
precision) without loss of generality, assume that all jobs processing
times (scaling  by $\epsilon^{-1}$) are integers  and upper-bounded by
$P$.

\paragraph{Relationship between  \textsf{SPT} and \textsf{EQUI}}

\begin{lemma}		\label{lem:coor-sum}
For any $A \geq 0, p > 0$ and $k, N$ integer, it holds that
$$
(k+1)\sum_{t = 1}^{N} (A + tp)^{k} \geq N (A + Np)^{k} 
$$
\end{lemma}

\begin{lemma} \label{lem:EQUI-SPT} Let  $\vect{x}$ be an assignment of
  jobs  to  machines.  Then,  the  \textsf{SPT}  policy minimizes  the
  $\ell_{k}$-norm  of  job  completion  times  with  respect  to  this
  assignment among all feasible policies.  Moreover, the \textsf{EQUI}
  policy induces an objective value at most $(2k+2)^{1/k}$ times higher.
\end{lemma}


\subsection{Upper bounds of the PoA induced by \textsf{SPT} and \textsf{EQUI}}


\begin{theorem}
The PoA of \textsf{SPT} with respect to the $\ell_{k}$-norm of job completion times
is $O(k)$. 
\end{theorem}
\begin{proof}
Let $\vect{x}$ and $\vect{x}^{*}$ be two arbitrary profiles. 
We focus on a machine $i$. Let $n_{1}, \ldots, n_{P}$ be the numbers of jobs in $\vect{x}$ which are assigned to
machine $i$ and have processing times $1, \ldots, P$,  respectively. Similarly, 
$m_{1}, \ldots, m_{P}$ are defined for profile $\vect{x}^{*}$.
Note that $n_{a}$ and $m_{a}$ are non-negative   for $1 \leq a \leq P$.
Applying Lemma~\ref{lem:smooth-SPT} for non-negative sequences 
$(n_{a})_{a=1}^{P}$, $(m_{a})_{a=1}^{P}$ and the positive increasing 
sequence $(a)_{a=1}^{P}$, we have:
\begin{align*}
\sum_{a=1}^{P} & \left[\left(\sum_{b=1}^{a} b n_{b} + a \right)^{k} +
		\left(\sum_{b=1}^{a} b n_{b} + 2a \right)^{k} + \ldots + \left(\sum_{b=1}^{a} b n_{b} + m_{a} \cdot a \right)^{k} \right] \\
\leq & ~ \frac{k+1}{k+2} \cdot \sum_{a=1}^{P} \left[\left(\sum_{b=1}^{a-1} b n_{b} + a \right)^{k} +
		\left(\sum_{b=1}^{a-1} b n_{b} + 2a \right)^{k} + \ldots + \left(\sum_{b=1}^{a-1} b n_{b} + n_{a} \cdot a \right)^{k} \right] + \\
	& +  \Theta\left(\alpha^{k}(k+1)^{k}\right) \cdot  \left[\left(\sum_{b=1}^{a-1} b m_{b} + a \right)^{k} +
		\left(\sum_{b=1}^{a-1} b m_{b} + 2a \right)^{k} + \ldots + \left(\sum_{b=1}^{a-1} b m_{b} + m_{a} \cdot a \right)^{k} \right] 
\end{align*}
where $\alpha$ is a constant.

Observe that, by  definition of the cost under  the \textsf{SPT} policy
(see  page~\pageref{def:SPT}), 
the left-hand  side (of  the inequality
above)   is   an   upper   bound   for  $\sum_{j:   x^{*}_{j}   =   i}
c^{k}_{j}(x_{-j}, x^{*}_{j})$, while the right-hand side is exactly $\frac{k+1}{k+2} \cdot \sum_{j:  x_{j} = i} c^{k}_{j}(\vect{x}) +
\Theta\left(\alpha^{k}(k+1)^{k}\right)  \cdot \sum_{j: x^{*}_{j}  = i}
c^{k}_{j}(\vect{x}^{*}) $. Thus,
$$
\sum_{j: x^{*}_{j} = i} c^{k}_{j}(x_{-j}, x^{*}_{j}) 
	\leq \frac{k+1}{k+2} \cdot \sum_{j: x_{j} = i} c^{k}_{j}(\vect{x}) 
		+  \Theta\left(\alpha^{k}(k+1)^{k}\right) \cdot \sum_{j: x^{*}_{j} = i} c^{k}_{j}(\vect{x}^{*})
$$
As the inequality above holds for every machine $i$, summing over all machines we have:
$$
\sum_{j}  c^{k}_{j}(x_{-j}, x^{*}_{j}) 
	\leq \frac{k+1}{k+2} \cdot \sum_{j} c^{k}_{j}(\vect{x}) 
		+  \Theta\left(\alpha^{k}(k+1)^{k}\right) \cdot \sum_{j} c^{k}_{j}(\vect{x}^{*})
$$
By the smooth argument, $C^{k}(\vect{x}) \leq \left(\alpha^{k}(k+1)^{k+1}\right) C^{k}(\vect{x}^{*})$, i.e., 
$C(\vect{x}) \leq O(k) C(\vect{x}^{*})$. Moreover,
by Lemma~\ref{lem:EQUI-SPT}, the optimal schedule for any assignment 
could be done using the \textsf{SPT} policy. Therefore, the PoA is $O(k)$. 
\end{proof}

\begin{theorem}
The PoA of \textsf{EQUI} with respect to the $\ell_{k}$-norm of job completion times
is $O(2^{k})$.
\end{theorem}
\begin{proof}
  Let $\vect{x}$ and $\vect{x}^{*}$ be two arbitrary profiles.  We
  focus on a fixed machine $i$. Let $n_{1}, \ldots, n_{P}$ be the
  numbers of jobs in $\vect{x}$ which are assigned to machine $i$ and
  have processing times $1, \ldots, P$ respectively. Similarly,
  $m_{1}, \ldots, m_{P}$ are defined for profile $\vect{x}^{*}$ and
  machine $i$. Remark that $n_{a}$ and $m_{a}$
are non-negative for $1 \leq a \leq P$.
  By definition of cost under policy \textsf{EQUI} (see
  page \pageref{def:EQUI}), the cost of job $j$ assigned to machine $i$ in profile $\vect{x}$ can also express as
$
  c_{j}(\vect{x}) = \sum_{b=1}^{p_{i,j}}bn_b + p_{i,j}\sum_{b=p_{i,j}+1}^{P}n_b
$.
Applying Lemma~\ref{lem:smooth-EQUI} for non-negative sequences 
$(n_{a})_{a=1}^{P}$, $(m_{a})_{a=1}^{P}$ and the positive increasing 
sequence $(a)_{a=1}^{P}$, we have:
\begin{align*}
\sum_{a=1}^{P} & m_{a}\left(\sum_{b=1}^{a} b n_{b} + a \cdot \sum_{b=a+1}^{P} n_{b}+ a m_{a} \right)^{k} \\
&\leq \frac{k+1}{k+2} \cdot \sum_{a=1}^{P}n_{a}\left(\sum_{b=1}^{a} b n_{b} +  a \sum_{b=a+1}^{P} n_{b} \right)^{k}  
	+  \Theta\left(\alpha^{k}2^{(k+1)^{2}}\right) \cdot \sum_{a=1}^{P}m_{a}\left(\sum_{b=1}^{a} b m_{b} +  a \sum_{b=a+1}^{P} m_{b} \right)^{k} 
\end{align*}
where $\alpha$ is a constant. Therefore, we deduce that
$$
\sum_{j: x^{*}_{j} = i} c^{k}_{j}(x_{-j}, x^{*}_{j}) 
	\leq \frac{k+1}{k+2} \cdot \sum_{j: x_{j} = i} c^{k}_{j}(\vect{x}) 
		+  \Theta\left(\alpha^{k}2^{(k+1)^{2}}\right) \cdot \sum_{j: x^{*}_{j} = i} c^{k}_{j}(\vect{x}^{*})
$$
As the inequality above holds for every machine $i$, summing over all machines we have:
$$
\sum_{j}  c^{k}_{j}(x_{-j}, x^{*}_{j}) 
	\leq \frac{k+1}{k+2} \cdot \sum_{j} c^{k}_{j}(\vect{x}) 
		+   \Theta\left(\alpha^{k}2^{(k+1)^{2}}\right) \cdot \sum_{j} c^{k}_{j}(\vect{x}^{*})
$$
By smooth argument, $C^{k}(\vect{x}) \leq O\left(\alpha^{k}(k+2)2^{(k+1)^{2}}\right) C^{k}(\vect{x}^{*})$, 
i.e., $C(\vect{x}) \leq O(2^{k+1}) C(\vect{x}^{*})$. Moreover,
by Lemma~\ref{lem:EQUI-SPT}, for any assignment profile the \textsf{EQUI} policy 
induces social cost within $(2k+2)^{1/k}$ times the optimal schedule on the assignment 
according to the $\ell_{k}$-norm. Hence, the PoA is $O((2k+2)^{1/k}2^{k+1}) = O(2^{k})$.
\end{proof}


\subsection{Lower bounds of the PoA}

\begin{lemma}		\label{lem:EQUI-lowerbound}
The PoA of \textsf{EQUI} with respect to the $\ell_{k}$-norm of job completion times
is $\Omega(k/\log k)$.
\end{lemma}
\begin{proof}
The construction is the same as in \cite{CohenDurrThang:Non-clairvoyant-Scheduling}. Let $m$ be an integer. 
Define $n_j := \frac{2 (m-1)!}{(j-1)!}$ for $1 \leq j \leq m$ and $n := \sum_{j=1}^m n_j$. We consider the set of $m$ machines and $m$ groups of jobs $J_1, J_2, \ldots, J_m$. In group $J_j$ $(1 \leq j \leq m-1)$, there are $n_j$ jobs that can be scheduled on machine $j$ or $j+1$. Each job in group $J_j$ $(1 \leq j \leq m-1)$ has processing time $p_{jj} = \frac{(j-1)!}{(m-1)!} = \frac{2}{n_j}$ on machine $j$ and has processing time $p_{j+1,j} = \frac{j!}{2(m-1)!} = \frac{1}{n_{j+1}}$ on machine $j+1$. The last group ($J_m$) which has a single job that can be only scheduled on machine $m$.  This job  has processing time $p_{mm} = 1$ on machine $m$. 

Consider the strategy profile $\vect{x}$ in which half of the jobs in $J_j$ $(1 \leq j \leq m-1)$ are scheduled on machine $j$ and the other half are scheduled on machine $j+1$ (the job in $J_m$ are scheduled on machine $m$). This strategy profile is a strong Nash 
equilibrium~(see \cite{CohenDurrThang:Non-clairvoyant-Scheduling} for the proof). Observe that the cost of jobs in group $J_{j}$ for all $1 \leq j \leq m$ are the same and equals $\frac{n_{j-1}}{2}p_{j,j-1} + \frac{n_j}{2}p_{jj} = \frac{j-1}{2} + 1 = \frac{j+1}{2}$. Hence, the social cost $C(\vect{x})$ satisfies
\begin{align*}
C^{k}(\vect{x}) = \sum_{j=1}^{m} n_{j}c^{k}_{j} = \sum_{j=1}^{m} \frac{2 (m-1)!}{(j-1)!} \cdot \left(\frac{j+1}{2}\right)^{k}
	> \frac{(m-1)!}{2^{k-1}} \sum_{j=0}^{m} \frac{j^{k}}{j!}
\end{align*}

Let $B_{k}$ is the $k^{th}$ Bell number. Remark that $B_{k} = \Theta\left( (k/\log k)^{k}\right)$.
By a property of Bell polynomial~\cite[page 66]{Roman:The-Umbral-Calculus}, 
$eB_{k} = \sum_{j=0}^{\infty} \frac{j^{k}}{j!}$ (Dobi\'nski's formula). 
Hence for $m$ large enough, we have
$$
C^{k}(\vect{x}) > \frac{(m-1)!}{2^{k-1}} B_{k}
$$

Consider a profile $\vect{x}^{*}$ in which jobs in group $J_j$ ($1 \leq j \leq m$) are assigned to machine $j$.
In this profile, every job has cost 2 except the job in $J_{m}$ with cost 1. Thus, the social cost $C(\vect{x}^{*})$
satisfies
$$
C^{k}(\vect{x}^{*}) = 1 + \sum_{j=1}^{m-1} n_{j}2^{k} < 2^{k} \sum_{j=1}^{m} \frac{2(m-1)!}{j!} < e2^{k+1}(m-1)!
$$ 
Therefore, the PoA is at least $\left(\frac{B_{k}}{e \cdot 4^{k}}\right)^{1/k} = \Omega(k/\log k)$.
\end{proof}

\begin{theorem}
The PoA of any deterministic non-preemptive non-waiting strongly-local policy 
is $\Omega(k/\log k )$ with respect to the $\ell_{k}$-norm of job completion times.
\end{theorem}
\begin{proof}
Using the technique described in \cite{ColeCorreaGkatzelis:Inner-Product}, 
it is sufficient to prove that the PoA of \textsf{SPT} is $\Omega(k/\log k)$.
Consider the same construction in Lemma~\ref{lem:EQUI-lowerbound}. The only difference is that 
now we partition each group $J_{j}$ into $J^{1}_{j} \cup J^{2}_{j}$ where $|J^{1}_{j}| = |J^{2}_{j}| = n_{j}/2$.
In machine $j$, jobs in sub-group $J^{1}_{j}$ have higher priority in than the ones in $J^{2}_{j}$.
Inversely, in machine $j+1$,  jobs in sub-group $J^{1}_{j}$ have higher priority in than the ones in $J^{2}_{j}$.

Consider the profile in which jobs in $J^{1}_j$ and $J^{2}_{j}$ $(1 \leq j \leq m-1)$ 
are assigned to machine $j$ and $j+1$ respectively (the job in $J_m$ is assigned to machine $m$).
By the definition of priority and Lemma~\ref{lem:EQUI-lowerbound}, this profile is an equilibrium 
(under \textsf{SPT}). Besides, by Lemma~\ref{lem:EQUI-SPT}, the social cost of this profile under \textsf{SPT}
is within $(2k+2)^{1/k}$-fraction of that induced by \textsf{EQUI}. Notice that the optimal social cost under \textsf{SPT} 
is always upper-bounded by that under \textsf{EQUI}. Therefore, the PoA of \textsf{SPT} is 
$\Omega((2k+2)^{-1/k} \cdot k/\log k) = \Omega(k/\log k)$.
\end{proof}


\section{$\ell_{\infty}$-norms of Completion Times} 	\label{sec:makespan}

We consider local policies for the makespan social cost.
First, we revisit the policies \textsf{BCOORD} and \textsf{CCOORD} 
in \cite{Caragiannis:Efficient-coordination-mechanisms}
by giving simpler proofs in a unified manner that is based on the smooth argument 
and the smooth inequalities. With the highlight of this approach, 
we derive a policy \textsf{Balance} that 
gives the currently best performance among anonymous local policies
which always admits a Nash equilibrium. 

For any profile $\vect{x}$, the social cost $C(\vect{x}) = \max_{j}
c_{j}$.  Let $\vect{x}(i) = \{j: x_{j} = i\}$ be the set of jobs
assigned to machine $i$.  Define $L(\vect{x}(i)) := \sum_{j: x_{j} =i} p_{ij}$, 
$L(\vect{x}) := \max_{i} L(\vect{x}(i))$ for all machines
$1 \leq i \leq m$. Note that in an optimal assignment $\vect{x}^{*}$,
$C(\vect{x}^{*}) = L(\vect{x}^{*})$. For each job $j$, denote $q_{j}
:= \min\{p_{ij}: 1 \leq i \leq m\}$ and define $\rho_{ij} :=
p_{ij}/q_{j}$ for all $i,j$. Moreover
the following lemma
guarantees that the restriction to $m$-efficient assignment is as
efficient as in the case up to a constant.

\begin{lemma}[\cite{Caragiannis:Efficient-coordination-mechanisms}]	\label{m-efficient}
\label{lem:m-efficient}
Let $\vect{y}^{*}$ be an optimal assignment. Then, there exits a $m$-efficient 
assignment $\vect{x}^{*}$ such that $L(\vect{x}^{*}) \leq 2L(\vect{y}^{*})$.
\end{lemma} 


\subsection{Policy \textsf{BCOORD}, Revisited}
Let $k$ be a positive integer.
Under policy \textsf{BCOORD}~\cite{Caragiannis:Efficient-coordination-mechanisms}, 
in profile $\vect{x}$ in which job $j$ chooses machine $i$,
the completion time $c_{j}$ of $j$ equals $\rho_{ij}^{1/k} L(\vect{x}(i))$ if 
$\rho_{ij} \leq m$ and equals $\infty$ otherwise. 
As $\rho_{ij} \geq 1$, $c_{j} \geq L(\vect{x}(i))$ for all jobs $j$ 
assigned to machine $i$. So, the schedule of such jobs is feasible.
Note that the game under the policy does not always possess a 
Nash equilibrium \cite{Caragiannis:Efficient-coordination-mechanisms}.

\begin{lemma}
Let $\vect{x}$ and $\vect{x}^{*}$ be  an equilibrium and an arbitrary $m$-efficient profile, respectively. Then, it holds that
$\sum_{i} \sum_{j: x_{j} = i} q_{j} c^{k}_{j}(\vect{x}) 
	\leq O\left(k \alpha^{k} \left(\frac{k}{\log k}\right)^{k-1}\right) \sum_{i} \sum_{j: x^{*}_{j} = i} q_{j}c^{k}_{j}(\vect{x}^{*})$
where $\alpha$ is a constant. 
\end{lemma}

\begin{theorem}[\cite{Caragiannis:Efficient-coordination-mechanisms}]
The PoA of the game under policy $\textsf{BCOORD}$ is
$O\left(\frac{\log m}{\log\log m}\right)$ by choosing $k = \log m$.
\end{theorem}


\subsection{Policy \textsf{CCOORD}, Revisited}
For any integer $k$ and any non-empty set $A= \{a_{1}, \ldots, a_{n}\}$ of non-negative
reals. The function $\Psi_{k}$ is defined as the following
$$
\Psi_{k}(A) = k! \sum_{1 \leq d_{1} \leq \ldots \leq d_{k} \leq n} \prod_{t=1}^{k} a_{d_{t}}
$$

By an abuse of notation, we define   $L(A)= \sum_{i=1}^{n}  a_{i}$. 
Note that $\Psi_{1}(A) = L(A)$.
Under policy \textsf{CCOORD}~\cite{Caragiannis:Efficient-coordination-mechanisms}, 
in profile $\vect{x}$ in which job $j$ chooses machine $i$,
the completion time $c_{j}$ of $j$ equals $\left(\rho_{ij} \Psi_k(\vect{x}(i)) \right)^{1/k}$
if $\rho_{ij} \leq m$ and equals $\infty$ otherwise.
The game under policy \textsf{CCOORD} always admits a Nash 
equilibrium~\cite{Caragiannis:Efficient-coordination-mechanisms}.

\begin{lemma}		\label{lem:ccoord}
Let $\vect{x}$ be a Nash equilibrium. Then, for any $m$-efficient profile $\vect{x}^{*}$, it holds that
$\sum_{i}\sum_{j: x_{j} = i} q_{j}c^{k}_{j}(\vect{x}) 
	\leq  2^{k+1}(k+1)^{k+2} \sum_{i}\sum_{j: x^{*}_{j} = i} q_{j} c^{k}_{j}(\vect{x}^{*})$ 
\end{lemma}

\begin{theorem}[\cite{Caragiannis:Efficient-coordination-mechanisms}]
The PoA of the game under policy $\textsf{CCOORD}$ is
$O(\log^{2}m)$ by choosing $k = \log m$.
\end{theorem}


\subsection{Policy \textsf{Balance}}

Let $\vect{x}$ be a strategy profile. Let $\prec_{i}$ be a total order
on the  jobs assigned to machine  $i$.  Formally, $j  \prec_{i} j'$ if
$p_{ij}  < p_{ij'}$, or  $p_{ij}=p_{ij'}$ and  $j$ is priority over $j'$ (machine $i$
chooses a local preference over jobs based on their local identities to break ties).
Note that the policy does not need a global job identities 
(there is no communication cost between machines about job identities)
and a job may have different priority on different machines.
The policy is clearly anonymous.

The cost $c_{j}$ of job $j$ assigned to machine $i$ is defined as follows.
$$
c^{k}_{j}(\vect{x}) = 
\begin{cases}
	\frac{1}{q_{j}}\biggl[\Big(p_{ij} + \displaystyle\sum_{\substack{j': j' \prec_{i} j \\ x_{j'} = i}} p_{ij'}\Big)^{k+1} - 
		\Big(\displaystyle \sum_{\substack{j': j' \prec_{i} j \\ x_{j'} = i}} p_{ij'}\Big)^{k+1} \biggl] &\text{ if } \rho_{ij} \leq m,\\
	\infty &\text{ otherwise.}
\end{cases}
$$ 

Observe that the cost $c_{j}(\vect{x})$ of job $j$ satisfies 
\begin{align*}
c^{k}_{j}(\vect{x}) 
&\geq 
\frac{1}{q_{j}}\left[\Big(p_{ij} + \displaystyle\sum_{j': j' \prec_{i} j , \ x_{j'} = i} p_{ij'}\Big)^{k+1} - 
		\Big(\displaystyle \sum_{j': j' \prec_{i} j , \ x_{j'} = i} p_{ij'}\Big)^{k+1} \right] \\
&\geq
\frac{p_{ij}}{q_{j}}\Big(p_{ij} + \sum_{j': j' \prec_{i} j , \ x_{j'} = i} p_{ij'}\Big)^{k}  
\geq \Big(p_{ij}  + \sum_{j': j' \prec_{i} j , \ x_{j'} = i}p_{ij'}\Big)^{k}
\end{align*}
since $p_{ij}/q_{j} \geq 1$.   
As that holds for every job $j$
assigned to machine $i$, policy \textsf{Balance} is feasible.

Remark that even there is some similarity in the definition of \textsf{Balance} and policy
\textsf{ACCORD} \cite{Caragiannis:Efficient-coordination-mechanisms}, 
the latter is not anonymous. \textsf{ACCORD} uses a global
job ordering in its definition and it makes use this order to prove the 
existence and inefficiency of Nash equilibria whereas \textsf{Balance} uses only
local job identities in case of tie break (that is unavoidable for any policy).  

\begin{lemma}
The best-response dynamic under the \textsf{Balance} policy converges to a Nash equilibrium.
\end{lemma}
\begin{proof}
By the definition of the policy, any job $j$ will choose a machine $i$ such that 
$\rho_{ij} \leq m$. Moreover, since $q_{j}$ is fixed for each job $j$, the behavior of 
jobs is similar to that in the game in which the set of strategy of a player $j$
is the same as in the former except for machines $i$ with $\rho_{ij} > m$. 
Moreover, in the new game, player $j$ in profile $\vect{x}$ has cost $c'_{j}(\vect{x})$ 
such that
$$
\Big(c'_{j}(\vect{x})\Big)^{k} = \Big(p_{ij} + \sum_{j' \prec_{i} j} p_{ij'}\Big)^{k+1} - 
		\Big(\sum_{j' \prec_{i} j} p_{ij'}\Big)^{k+1}
$$
Hence, it is sufficient to prove that the better-response dynamic in the new game always converges.
The argument is the same as the one to prove the existence of Nash equilibrium for 
policy \textsf{SPT}~\cite{ImmorlicaLiMirrokni:Coordination-mechanisms}. 
Here we present a proof based on a geometrical approach.

First, define $\texttt{pos}_{i}(j) := 1 +  |\{j': j' \prec_{i} j, 1\leq j' \neq j \leq n\}|$ which represents 
the priority of job $j$ on machine $i$. For a
value $u \in \mathbb{R}^{+}$ and a job index $1 \leq t \leq n$, we associate to every profile $\vect x$
the quantity
\[
    |\vect x|_{u,t} := |  \{ j: c'_j(\vect x)<u \text{ or }  c'_j(\vect x)=u, \texttt{pos}_{x_{j}}(j)\leq t \}|.
\]
We use  it to  define a partial  order $\prec$ on  profiles.  Formally
$\vect x \prec \vect y$ if for the lexicographically smallest pair
$(u,t)$ such that $|\vect x|_{u,t} \neq |\vect y|_{u,t}$ we have
$|\vect x|_{u,t} < |\vect y|_{u,t}$.
\begin{figure}[htbp]
    \centering{
      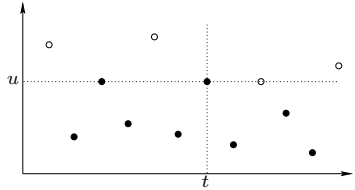
     }
    \caption{An geometrical  illustration of $|\vect  x|_{u,t}$, every
      dot is a $(j,c_j(\vect x))$ pair, colored black if counted in $|\vect  x|_{u,t}$.}
    \label{fig:balance}
\end{figure}

We show that the profile strictly increases according to this order, whenever a
job changes to another machine while decreasing its cost. Let $j$ be
such a job changing from  machine $a$ in profile $\vect{x}$ to machine
$b$,   resulting   in   a    profile   $\vect{y}$.    We   know   that
$c'_{j}(\vect{y})  < c'_{j}(\vect{x})$.   Remark that  only  jobs $j'$
with $x_{j'}=b$ might have the cost in $\vect y$ larger than that in $\vect x$ (by
definition of  the cost  $c'$).  Moreover, such job $j'$
with $x_{j'}=b$ and $j'$ has a different costs in $\vect{x}$ and $\vect{y}$, 
it must be $j \prec_{b} j'$,  which also implies  $c'_{j'}(\vect x)\geq c'_{j}(\vect y)$.   
In the  same spirit, some jobs $j'$  with $x_{j'}=a$  might decrease their  cost, 
but  not below $c'_j(\vect x)$. 

Consider $u=c'_j(\vect y)$  and $t= \texttt{pos}_{b}(j)$. 
We have that $|\vect x|_{u',t'} = |\vect y|_{u',t'}$
for all $u' < u$ and all $t'$. If job $j$ is the only job with processing time 
$p_{bj}$ among the ones $\{j': x_{j'} = b\}$, then 
$|\vect y|_{u,t} = |\vect x|_{u,t} + 1$. Otherwise,
$|\vect y|_{u,t'} = |\vect x|_{u,t'}$ for $t' < t$ and  
$|\vect y|_{u,t} = |\vect x|_{u,t} + 1$.

Therefore $(u,t)$ is the first  lexicographical  pair  
where $|\vect x|_{u,t} \neq |\vect y|_{u,t}$ and 
$|\vect y|_{u,t} > |\vect x|_{u,t}$.
Hence,  since the  set of  strategy profiles  is finite,  the better-response
dynamic must converge to a pure Nash equilibrium.  This completes the proof.
\end{proof}

Remark that the game under \textsf{Balance} convergences fast to Nash equilibria 
in the best-response dynamic (the argument is the same as 
\cite[Theorem 12]{ImmorlicaLiMirrokni:Coordination-mechanisms}).

\begin{lemma}		\label{lem:balance-smooth}
Let $\vect{x}$ and $\vect{x}^{*}$ be an equilibrium and an $m$-efficient arbitrary profile, respectively.
Then, $\sum_{i=1}^{m} L^{k+1}(\vect{x}(i)) \leq O(\alpha^{k} k^{k+1})\sum_{i=1}^{m} L^{k+1}(\vect{x}^{*}(i))$
where $\alpha$ is some constant.
\end{lemma}
\begin{proof}
We focus on an arbitrary job $j$. Denote $i = x_{j}$ and $i^{*} = x^{*}_{j}$. As $\vect{x}$ is an equilibrium, we have
$c^{k}_{j}(\vect{x}) \leq c^{k}_{j}(x_{-j},x^{*}_{j})$, i.e, 
\begin{align}	\label{ineq:balance}
&\Big(p_{ij} + \sum_{\substack{j': j' \prec_{i} j \\ x_{j'} = i}} p_{ij'}\Big)^{k+1} 
	- \Big(\sum_{\substack{j': j' \prec_{i} j \\ x_{j'} = i}} p_{ij'}\Big)^{k+1}
\leq \Big(p_{i^{*}j} + \sum_{\substack{j': j' \prec_{i^{*}} j \\ x_{j'} = i^{*}}} p_{i^{*}j'}\Big)^{k+1} 
		- \Big(\sum_{\substack{j': j' \prec_{i^{*}} j \\ x_{j'} = i^{*}}} p_{i^{*}j'}\Big)^{k+1} \notag \\
& \qquad \leq \Big( p_{i^{*}j} + L(\vect{x}(i^{*}))\Big)^{k+1} -  \Big(L(\vect{x}(i^{*}))\Big)^{k+1} 
	\leq (k+1)p_{i^{*}j} \Big( p_{i^{*}j} + L(\vect{x}(i^{*}))\Big)^{k}
\end{align}
where the second inequality is due to the fact that  $(z+a)^{k+1} - z^{k+1}$ is an increasing in $z$ 
(for $a > 0$) and $\sum_{\substack{j': j' \prec_{i} j \\ x_{j'} = i^{*}}} p_{i^{*}j'} \leq L(\vect{x}(i^{*}))$;
the third inequality is due to Lemma~\ref{lem:simple} (by dividing both sides by 
$(p_{i^{*}j} + L(\vect{x}(i^{*})))^{k+1}$ and applying $z = \frac{p_{i^{*}j}}{p_{i^{*}j} + L(\vect{x}(i^{*}))}$
in the statement of Lemma~\ref{lem:simple}). Therefore,
\begin{align*}
\sum_{i=1}^{m} & L^{k+1}(\vect{x}(i)) = \sum_{i=1}^{m}\sum_{j:x_{j} = i} q_{j}c^{k}_{j}(\vect{x})
	\leq \sum_{i=1}^{m}\sum_{j:x_{j} = i} q_{j}c^{k}_{j}(x_{-j},x^{*}_{j}) \\
&\leq \sum_{i=1}^{m}\sum_{j:x^{*}_{j} = i} (k+1)p_{ij} \Big( p_{ij} + L(\vect{x}(i))\Big)^{k+1}
	\leq (k+1)\sum_{i=1}^{m} L(\vect{x}^{*}(i))\Big( L(\vect{x}(i)) + L(\vect{x}^{*}(i))\Big)^{k} \\
&\leq (k+1)\sum_{i=1}^{m} \frac{k}{(k+1)^{2}} L^{k+1}(\vect{x}(i)) 
		+ O\left(\alpha^{k}k^{k-1}\right) L^{k+1}(\vect{x}^{*}(i))
\end{align*}
where the first inequality is because $\vect{x}$ is an equilibrium; 
the second inequality is due to the sum of Inequality~(\ref{ineq:balance}) taken over all jobs $j$;
and the fourth inequality is due to case (\ref{b2}) of Lemma~\ref{lem:smooth-simple}. 
Arranging the terms, the lemma follows.
\end{proof}

\begin{theorem}	\label{thm:balance}
The PoA of the game under policy $\textsf{Balance}$ is at most
$O(\log m)$ by choosing $k = \log m$. 
\end{theorem}
\begin{proof}
Let $\vect{y}^{*}$ be an optimal assignment and $\vect{x}^{*}$ be an $m$-efficient 
assignment with property of Lemma~\ref{lem:m-efficient}. Let $\vect{x}$ be an 
equilibrium. Remark that $\vect{x}$ is a $m$-efficient assignment since
every job can always get a bounded cost.  
Consider a job $j$ assigned to machine $i$ in profile $\vect{x}$.
As $\vect{x}$ is a $m$-efficient assignment, by the definition of the policy \textsf{Balance} 
\begin{align*}
c^{k}_{j}(\vect{x}) &= \frac{1}{q_{j}}\biggl[\Big(p_{ij} + \displaystyle\sum_{\substack{j': j' \prec_{i} j \\ x_{j'} = i}} p_{ij'}\Big)^{k+1} - 
		\Big(\displaystyle \sum_{\substack{j': j' \prec_{i} j \\ x_{j'} = i}} p_{ij'}\Big)^{k+1} \biggl] \\
		& \leq \frac{1}{q_{j}}\biggl[\Big(L(\vect{x}(i))\Big)^{k+1} - 
		\Big( L(\vect{x}(i)) - p_{ij} \Big)^{k+1} \biggl]
			\leq (k+1)\rho_{ij}L^{k}(\vect{x}(i))
\end{align*}
where the first inequality is because function $(a+x)^{k+1} - x^{k+1}$ is increasing; and  
the last inequality is due to Lemma~\ref{lem:simple} (by dividing both sides by 
$L^{k+1}(\vect{x}(i))$ and applying $z = \frac{p_{ij}}{L(\vect{x}(i))}$
in the statement of Lemma~\ref{lem:simple}). 
Moreover, by Lemma~\ref{lem:balance-smooth}, we have
$$
L^{k+1}(\vect{x}) \leq \sum_{i=1}^{m}L^{k+1}(\vect{x}(i)) 
	\leq O(\alpha^{k} k^{k+1})\sum_{i=1}^{m} L^{k+1}(\vect{x}^{*}(i))
	\leq O(\alpha^{k} k^{k+1} m) L^{k+1}(\vect{x}^{*})
$$
for some constant $\alpha$. Therefore, 
\begin{align*}
C(\vect{x}) &= \max_{j} c_{j}(\vect{x}) \leq \max_{i,j} \Big((k+1)\rho_{ij}\Big)^{1/k} L(\vect{x}(i)) 
	\leq \Big((k+1)m\Big)^{1/k}L(\vect{x}) \\
&\leq O\left(\Big(k^{k+2} m^{2}\Big)^{1/k}\right) L(\vect{x}^{*}) 
	\leq O\left(\Big(k^{k+2} m^{2}\Big)^{1/k}\right) L(\vect{y}^{*})
\end{align*}
where the last inequality is due to Lemma~\ref{m-efficient}. 
Choosing $k = \log m$, the theorem follows. 
\end{proof}
%
%
\bibliographystyle{plainnat}
\bibliography{gametheory}

\newpage    

\section*{APPENDIX}
\setcounter{section}{1}

\section{Smooth inequalities}

Before proving Lemma~\ref{lem:smooth-EQUI}, we need to show the 
following lemma. The technique is similar to the proof of 
Lemma~\ref{lem:smooth-SPT}.

\setcounter{theorem}{4}
%
%
\begin{lemma}    \label{lem:smooth-SPT-modif}   
 For  any   non-negative
  sequences  $(n_{i})_{i=1}^{P}$,  $(m_{i})_{i=1}^{P}$,  and  for  any
  positive increasing  sequence $(q_{i})_{i=1}^{P}$, define  $A_{i} :=
  n_{1}q_{1}  + \ldots +  n_{i-1}q_{i-1} +  n_{i}q_{i}$ and  $B_{i} :=
  m_{1}q_{1} +  \ldots +  m_{i-1}q_{i-1} + m_{i}q_{i}$  for $1  \leq i
  \leq P$. Then, it holds that
$$
\sum_{i=1}^{P}m_{i}(A_{i}     +    m_{i}q_{i})^{k}     \leq    \mu_{k}
\sum_{i=1}^{P}  n_{i}A_{i}^{k}   +  \lambda_{k}  \sum_{i=1}^{P}  m_{i}
B_{i}^{k}
$$
where     $\mu_{k}      =     \frac{k+1}{k+2}$,     $\lambda_{k}     =
\Theta(\alpha^{k}(k+1)^{k})$  for  some  constant $\alpha$.   
\end{lemma}
\begin{proof}
  Denote  $r_{j}  =  1/q_{j}$,  $a_{j}  =  n_{j}q_{j}$  and  $b_{j}  =
  m_{j}q_{j}$ for $1 \leq j \leq P$.  So $A_{i} = \sum_{j=1}^{i}a_{j},
  B_{i} = \sum_{j=1}^{i}b_{j}$.  The inequality is equivalent to
$$
\sum_{i=1}^{P} r_{i} b_{i}(A_{i} + b_{i})^{k} 
\leq \mu_{k} \sum_{i=1}^{P} r_{i}a_{i} A^{k}_{i} 
  + \lambda_{k} \sum_{i=1}^{P} r_{i}b_{i} B^{k}_{i}
$$
For convenience we set $r_{P+1} = 0$. The inequality could be written as 
$$
\sum_{i=1}^{P} (r_{i} - r_{i+1}) \sum_{j=1}^{i} b_{j}(A_{j} + b_{j})^{k}
\leq \sum_{i=1}^{P} (r_{i} - r_{i+1}) 
	\left[\mu_{k} \sum_{j=1}^{i} a_{j} A^{k}_{j} + \lambda_{k} \sum_{j=1}^{i} b_{j} B^{k}_{j} \right]
$$

As $(r_{i})_{i=1}^{P}$  is a decreasing sequence (e.g.\  $r_{i} - r_{i+1}
\geq 0 ~\forall1  \leq i \leq P$), it is sufficient  to prove that for
all $1 \leq i \leq P$,
\begin{equation}	\label{ineq:EQUI-modif-bA}
\sum_{j=1}^{i} b_{j}(A_{j} + b_{j})^{k} 
\leq \left[ \mu_{k}\sum_{j=1}^{i} a_{j} A^{k}_{j} + \lambda_{k} \sum_{j=1}^{i} b_{j} B^{k}_{j} \right]
\end{equation}
It remains to show Inequality~(\ref{ineq:EQUI-modif-bA}).

By Lemma~\ref{lem:simple}, we have 
$(k+1)a_{j}A^{k}_{j} \geq A^{k+1}_{j} - (A_{j} - a_{j})^{k+1}$ 
(since  dividing both sides by $A^{k+1}_{j}$, we obtain the inequality 
in Lemma~\ref{lem:simple} for $z = a_{j}/A_{j}$).
Therefore,
\begin{align*}
(k+1)\sum_{j=1}^{i} a_{j} A^{k}_{j} &\geq \sum_{j=1}^{i} \left[A^{k+1}_{j} - (A_{j} - a_{j})^{k+1}\right] 
	= A^{k+1}_{i} - (A_{1} - a_{1})^{k+1} 
	= A^{k+1}_{i}
\end{align*}
since $A_{j} = A_{j+1} - a_{j+1}$ for all $1 \leq j \leq i-1$ and $A_{1} = a_{1}$. Similarly,
$(k+1)\sum_{j=1}^{i} a_{j} B^{k}_{j} \geq B^{k+1}_{i}$.
Thus, to prove inequality (\ref{ineq:EQUI-modif-bA}), it is sufficient to prove that for all 
$1 \leq i \leq P$,
$$	
\sum_{j=1}^{i} b_{j}(A_{j}+b_{j})^{k} 
\leq \left( \frac{\mu_{k}}{k+1} A^{k+1}_{i} + \frac{\lambda_{k}}{k+1} B^{k+1}_{i} \right)
$$
Observe that
 $$
 \sum_{j=1}^{i} b_{j}(A_{j}+b_{j})^{k} \leq \sum_{j=1}^{i} b_{j} (A_{i}+B_{i})^{k}  \leq B_{i} (A_{i}+B_{i})^{k}
 $$ 
Hence, we only need to argue that 
\begin{equation}	\label{ineq:smooth-SPT-modif}
B_{i} (A_{i}+B_{i})^{k} \leq \left( \frac{\mu_{k}}{k+1} A^{k+1}_{i} + \frac{\lambda_{k}}{k+1} B^{k+1}_{i} \right)
\end{equation}

Choose $\mu_{k} = \frac{k+1}{k+2}$ and apply case~(\ref{b2}) of Lemma~\ref{lem:smooth-simple} 
(now  $a(k) = \frac{(k+1)}{k(k+2)}$  and $(k-1)a(k)$  is bounded  by a
constant),      we     deduce      that     by      $\lambda_{k}     =
\Theta(\alpha^{k}(k+1)^{k})$ for a constant $\alpha$, the
Inequality~(\ref{ineq:smooth-SPT-modif})  holds.   
\end{proof}

%
\newcommand{\AONE}{A^{(1)}}%
\newcommand{\ATWO}{A^{(2)}}%
\newcommand{\BONE}{B^{(1)}}%
\newcommand{\BTWO}{B^{(2)}}%
%

\setcounter{theorem}{3}

\begin{lemma}		
For any non-negative sequences $(n_{i})_{i=1}^{P}$, $(m_{i})_{i=1}^{P}$, and for any positive increasing
sequence $(q_{i})_{i=1}^{P}$, define $A_{i} = n_{1}q_{1} + \ldots + n_{i-1}q_{i-1} + (n_{i} + \ldots + n_{P})q_{i}$
and $B_{i} = m_{1}q_{1} + \ldots + m_{i-1}q_{i-1} + (m_{i} + \ldots + m_{P})q_{i}$ for $1 \leq i \leq P$.
Then, it holds that
$$
\sum_{i=1}^{P}m_{i}(A_{i} + m_{i}q_{i})^{k} 
\leq \mu_{k} \sum_{i=1}^{P} n_{i}A_{i}^{k} + \lambda_{k} \sum_{i=1}^{P} m_{i} B_{i}^{k}
$$
where
$\mu_{k} = \frac{k+1}{k+2}$, $\lambda_{k} = \Theta(\alpha^{k}2^{(k+1)^{2}})$
for some constant $\alpha$.  
\end{lemma}
\begin{proof}
Let $\AONE_{i} = n_{1}q_{1} + \ldots + n_{i-1}q_{i-1} + n_{i}q_{i}$, 
$\ATWO_{i} = (n_{i} + \ldots + n_{P})$
and $\BONE_{i} = m_{1}q_{1} + \ldots + m_{i-1}q_{i-1} + m_{i}q_{i}$, 
$\BTWO_{i} = (m_{i} + \ldots + m_{P})$ for $1 \leq i \leq P$.
So, by definition we have
$\AONE_{i} \leq A_{i}$, $\ATWO_{i} \leq A_{i}$ and $A_i \leq \AONE_{i}+\ATWO_{i}q_i$ for $1 \leq i \leq N$.
Similarly,  $\BONE_{i} \leq B_{i}$, $\BTWO_{i} \leq B_{i}$ and $B_i \leq \BONE_{i}+\BTWO_{i}q_i$ for $1 \leq i \leq P$.
By convention, let $\AONE_{0}=\BONE_{0}=0$.

Thus, we have, for all $1 \leq i \leq P$,
$$
(A_{i}+m_{i}q_{i})^{k} \leq \left[\left(\AONE_{i} + m_{i}q_{i}\right) + q_{i}\left(\ATWO_{i}+m_{i}\right) \right]^{k} %
\leq 2^{k}\left[\left(\AONE_{i} + m_{i}q_{i}\right)^{k} + q_{i}^{k}\left(\ATWO_{i} + m_{i}\right)^{k}\right]
$$
and  
$$
\left(\AONE_{i}\right)^{k} + q_{i}^{k} \left(\ATWO_{i}\right)^{k} \leq 2 A_{i}^{k},
\qquad \left(\BONE_{i}\right)^{k} + q_{i}^{k} \left(\BTWO_{i}\right)^{k} \leq 2 B_{i}^{k}.
$$ 
Therefore, to prove the lemma, it is sufficient to argue that 
\begin{equation}	\label{ineq:smooth-EQUI-bilinear1}
2^{k}\sum_{i=1}^{P}m_{i} \left(\AONE_{i} + m_{i}q_{i}\right)^{k} 
\leq  \frac{\mu_{k}}{2} \sum_{i=1}^{P} n_{i} \left(\AONE_{i}\right)^{k} 
	+ \frac{\lambda_{k}}{2} \sum_{i=1}^{P} m_{i} \left(\BONE_{i}\right)^{k}
\end{equation}
and
\begin{equation}		\label{ineq:smooth-EQUI-bilinear2}
2^{k} \sum_{i=1}^{P}q_{i}^{k}m_{i}\left(\ATWO_{i}+ m_{i}\right)^{k}
\leq  \frac{\mu_{k}}{2}  \sum_{i=1}^{P} q_{i}^{k}n_{i}\left(\ATWO_{i}\right)^{k} 
		+\frac{\lambda_{k}}{2} \sum_{i=1}^{P} q_{i}^{k}m_{i}  \left(\BTWO_{i}\right)^{k} 
\end{equation}

Inequality~(\ref{ineq:smooth-EQUI-bilinear1}) holds for $\mu_{k} = \frac{k+1}{k+2}$ and
$\lambda_{k} = \Theta(\alpha^{k} 2^{(k+1)^{2}})$ for some constant $\alpha$
(the proof is similar to that of Lemma~\ref{lem:smooth-SPT-modif} but in the end,
the case (\ref{b3}) of Lemma~\ref{lem:smooth-simple} is used to derive the value
of $\lambda_{k}$).

Consider       inequality~(\ref{ineq:smooth-EQUI-bilinear2}).      For
convenience set $q_{0} = 0$.
The inequality could be rewritten as
\begin{align*}
2^{k}\sum_{i=1}^{P} & (q_{i}^{k} - q_{i-1}^{k}) \sum_{j = i}^{P} m_{j}\left(\ATWO_{j} + m_{j}\right)^{k} \\
&\leq \frac{\mu_{k}}{2} \sum_{i=1}^{P} (q_{i}^{k} - q_{i-1}^{k}) \sum_{j = i}^{P} n_{j}\left(\ATWO_{j}\right)^{k} %
	+ \frac{\lambda_{k}}{2} \sum_{i=1}^{P}(q_{i}^{k} - q_{i-1}^{k}) \sum_{j = i}^{P} m_{j}\left(\BTWO_{j}\right)^{k} 
\end{align*}
As the sequence $(q_{i})_{i=1}^{P}$ is increasing, it is sufficient to prove that 
\begin{equation}	\label{ineq:smooth-EQUI-tail}
2^{k} \sum_{j = i}^{P} m_{j} \left(\ATWO_{j} + m_{j}\right)^{k} 
\leq \frac{\mu_{k}}{2} \sum_{j = i}^{P} n_{j}\left(\ATWO_{j}\right)^{k} 
		+ \frac{\lambda_{k}}{2} \sum_{j = i}^{P} m_{j}\left(\BTWO_{j}\right)^{k}
\end{equation}

The inequality above could be considered as a corollary of Inequality~(\ref{ineq:smooth-EQUI-bilinear1})
by rewriting the indices in backward ($i \mapsto P+1-i$) and considering the sequence $q_{i} = 1$ 
for $1 \leq i \leq P$. Precisely, fix an index $i$ and  
applying Inequality~(\ref{ineq:smooth-EQUI-bilinear1}) for the sequence
$q_{j} = 1$ for $1 \leq j \leq P+1-i$ and two sequences $(n'_{j})_{j=1}^{P+1-i}$, and $(m'_{j})_{j=1}^{P+1-i}$ 
defined as $n'_{j}= n_{ P+1-j}$ and $m'_{i}= m_{ P+1-j}$ for all  $1 \leq j \leq P+1-i$. 

$$
2^{k}\sum_{j=1}^{P+1-i}m'_{j} \left(\left(\sum_{t=1}^{j} n'_{t}\right)+ m'_{j}\right)^{k} 
\leq \frac{\mu_{k}}{2} \sum_{j=1}^{P+1-i} n'_{j} \left(\sum_{t=1}^{j} n'_{t}\right)^{k}
+ \frac{\lambda_{k}}{2} \sum_{j=1}^{P} m'_{j} \left(\sum_{t=1}^{j} m'_{t}\right)^{k}
$$

Replacing $n'_{j}$ by $n_{j}$ and $m'_{j}$ by $m_{j}$ for $1 \leq j \leq P+1-i$, we get 
Inequality~(\ref{ineq:smooth-EQUI-tail}). 
\end{proof}


\section{$\ell_{k}$-norms of Completion Times} 

\begin{lemma}		\label{lem:coor-sum}
For any $A \geq 0, p > 0$ and $k, N$ integer, it holds that
$$
(k+1)\sum_{t = 1}^{N} (A + tp)^{k} \geq N (A + Np)^{k} 
$$
\end{lemma}
\begin{proof}
First, for all $0\leq h \leq k$ we have 
$$
\sum_{t=1}^{N} t^{h} = N^{h+1}\sum_{t=1}^{N} \left(\frac{t}{N}\right)^{h} \frac{1}{N}
\geq N^{h+1} \int_{0}^{1} x^{h}dx = \frac{N^{h+1}}{h+1}
$$
where the inequality is because the function $x^{h}$ is increasing.
Thus, $(k+1)\sum_{t=1}^{N} t^{h} \geq N^{h+1}$
for all $0\leq h \leq k$. Therefore,
 \begin{equation*}	\label{ineq:Bernoulli}
 \binom{k}{h} A^{k-h} p^{h} (k+1)\sum_{t=1}^{N} t^{h} \geq  \binom{k}{h} N^{h+1}A^{k-h}p^{h}	\quad \forall 0 \leq h \leq k
 \end{equation*}
Summing the inequalities over $0 \leq h \leq k$, we obtain    
$$ 
(k+1)\sum_{t = 1}^{N} (A + tp)^{k}
= \sum_{h=0}^{k}\binom{k}{h} A^{k-h} p^{h} (k+1)\sum_{t=1}^{N} t^{h} 
\geq \sum_{h=0}^{k} \binom{k}{h} N^{h+1}A^{k-h}p^{h}
= N (A + Np)^{k} 
$$
\end{proof}

\begin{lemma} \label{lem:EQUI-SPT} Let  $\vect{x}$ be an assignment of
  jobs  to  machines.  Then,  the  \textsf{SPT}  policy minimizes  the
  $\ell_{k}$-norm  of  job  completion  times  with  respect  to  this
  assignment among all feasible policies.  Moreover, the \textsf{EQUI}
  policy induces an objective value at most $(2k+2)^{1/k}$ times higher.
\end{lemma}
\begin{proof}
  Consider a machine $i$ and let $N$ be the number of jobs assigned to
  $i$  by the  profile  $\vect{x}$.   These $N$  jobs  are renamed  in
  increasing  order of processing  times, and  since we  fixed machine
  $i$, for convenience we drop  index $i$ in the processing times.  So
  we denote the $N$ processing  times as $p_{1} \leq p_{2} \leq \ldots
  \leq p_{N}$.  In any schedule of those jobs, there exist distinct
  jobs with completion times at least $p_{1}, p_{1}  + p_{2}, \ldots,
  p_{1}  +  \ldots  +  p_{h}$.   Hence,  the  $\ell_{k}$-norm  on  the
  completion  times of such  jobs is  at least  $\left( \sum_{j=1}^{h}
    (p_{1} + \ldots +  p_{j})^{k} \right)^{1/k}$, which is attained by
  the \textsf{SPT} policy.

These jobs are partitioned into different classes where jobs in the same class have the same processing time, i.e., 
there are $n_{1}, \ldots, n_{h}$ jobs with processing times $q_{1}, \ldots, q_{h}$. 
In the following, we will prove by induction on the number of classes (i.e parameter $h$) that the objective value induced by \textsf{EQUI} policy is 
within $(2k+2)^{1/k}$ times the one induced by \textsf{SPT}, or  that is, 
\begin{align}	\label{ineq:EQUI-SPT}
(2k+2) \sum_{j=1}^{h} &\sum_{t=1}^{N}(n_{0}q_{0} + \ldots +  n_{j-1}q_{j-1} + t \cdot q_{j})^{k} \notag \\
&\geq 2 \sum_{j=1}^{h} n_{j}\Big(n_{0}q_{0} + n_{1}q_{1} + \ldots + n_{j-1}q_{j-1} 
						+ (n_{j} + \ldots + n_{h})q_{j}\Big)^{k} 
\end{align}
where for convenience we denote $n_{0} = 0$ and $q_0 = 0$.

Consider the basis case where all jobs have the same processing time ($h = 1$).  
When $h=1$, Inequality (\ref{ineq:EQUI-SPT}) is equivalent to
$$
2(k+1)q_{1}^{k}\sum_{t=1}^{n_{1}} t^{k} \geq n_{1}^{k+1}q_{1}^{k}
$$ 
which is straightforward by Lemma~\ref{lem:coor-sum}. 

Now, assume that Inequality (\ref{ineq:EQUI-SPT}) holds for $h$ classes of jobs. 
We will prove that this statement also holds for $(h+1)$ classes. 

Define function $q(z) := q_{h} + (q_{h+1} - q_{h})z$ for $0 \leq z \leq 1$. 
Then, all jobs with processing time $q_{h+1}$ could be seen as having processing time $q(1)$.
Define $g(z) := 2(k+1) \sum_{t=1}^{n_{h+1}}(A + t \cdot q(z))^{k} -
n_{h+1}(A+ n_{h+1}q(z))^{k} $ where $A = \sum_{j=0}^{n_{h}}
n_{j}q_{j}$. Consider function $f(z)$ as follows: 
\begin{align*}
f(z)=g(z)+2(k+1) & \sum_{j=1}^{h} \sum_{t=1}^{n_{j}}(n_{0}q_{0} + \ldots + n_{j-1}q_{j-1} + t \cdot q_{j})^{k} \\
	& - \sum_{j=1}^{h} n_{j}\Big(n_{0}q_{0} + n_{1}q_{1} + \ldots + n_{j-1}q_{j-1} 
						+ (n_{j} + \ldots + n_{h+1})q_{j}\Big)^{k}
\end{align*}
Inequality (\ref{ineq:EQUI-SPT}) is equivalent to prove that $f(1) \geq 0$. 
By the induction hypothesis, we have $f(0) \geq 0$ since for $z = 0$ there are exactly 
$h$ classes in the inequality (\ref{ineq:EQUI-SPT}).

Consider the derivative of $f(z)$.
$$
f'(z) = g'(z) = k(q_{h+1} - q_{h})\left[ 2(k+1)\sum_{t=1}^{n_{h+1}} t(A + t q(z))^{k-1} - n_{h+1}^{2}(A + n_{h+1}q(z))^{k-1}\right]
$$
By Chebyshev sum inequality on two increasing sequences $(t)_{1}^{n_{h}}$ and 
$\left((A + t q(z))^{k-1}\right)_{t=1}^{n_{h}}$, we have:
\begin{align*}
2(k+1) &\sum_{t=1}^{n_{h+1}} t(A + t q(z))^{k-1} 
	\geq \frac{2(k+1)}{n_{h+1}}\left(\sum_{t=1}^{n_{h+1}} t\right) 
			\cdot \left( \sum_{t=1}^{n_{h+1}} \left(A + t q(z)\right)^{k-1} \right) \\
	&>  n_{h+1} (k+1) \sum_{t=1}^{n_{h+1}} \left(A + t q(z)\right)^{k-1} 
	>  n_{h+1}^{2}(A + n_{h+1}q(z))^{k-1}
\end{align*}
where the last inequality is due to Lemma~\ref{lem:coor-sum}. Thus, $f'(z) = g'(z) > 0$.
Hence, $f(1) \geq f(0) \geq 0$. 
\end{proof}


\section{$\ell_{\infty}$-norms of Completion Times} 	\label{sec:makespan}


\subsection{Policy \textsf{BCOORD}, Revisited}

\setcounter{theorem}{1}

\begin{lemma}
Let $\vect{x}$ and $\vect{x}^{*}$ be  an equilibrium and an arbitrary $m$-efficient profile, respectively. Then, it holds that
$\sum_{i} \sum_{j: x_{j} = i} q_{j} c^{k}_{j}(\vect{x}) 
	\leq O\left(k \alpha^{k} \left(\frac{k}{\log k}\right)^{k-1}\right) \sum_{i} \sum_{j: x^{*}_{j} = i} q_{j}c^{k}_{j}(\vect{x}^{*})$
where $\alpha$ is a constant. 
\end{lemma}
\begin{proof}
Let $j$ be an arbitrary job. 
Since $\vect{x}$ is an equilibrium, $c_{j}(\vect{x}) \leq c_{j}(x_{-j},x^{*}_{j})$, 
thus $q_{j}c^{k}_{j}(\vect{x}) \leq q_{j}c^{k}_{j}(x_{-j},x^{*}_{j})$ for all jobs $j$.
Using the smooth argument, in order to prove the lemma, it is sufficient to argue that
\begin{equation}	\label{ineq:BCOORD}
\sum_{i}\sum_{j: x^{*}_{j} = i} q_{j} c^{k}_{j}(x_{-j},x^{*}_{j}) 
\leq \frac{k}{k+1}\sum_{i}\sum_{j: x_{j} = i} q_{j} c^{k}_{j}(\vect{x})
	+ O\left(\alpha^{k}\left(\frac{k}{\log k}\right)^{k-1}\right)\sum_{i}\sum_{j: x^{*}_{j} = i} q_{j} c^{k}_{j}(\vect{x}^{*})
\end{equation}
for some constant $\alpha$.

Applying Lemma~\ref{lem:smooth-simple}(case (\ref{b1})), for any machine $i$ we have 
$$
\quad L(\vect{x}^{*}(i)) \left( L(\vect{x}(i)) + L(\vect{x}^{*}(i)) \right)^{k}
\leq \frac{k}{k+1}L^{k+1}(\vect{x}(i))
	+ O\left(\alpha^{k} \left(\frac{k}{\log k}\right)^{k-1}\right)L^{k+1}(\vect{x}^{*}(i))
$$
for some constant $\alpha$. Moreover, by definition of \textsf{BCOORD}, 
$\sum_{j: x^{*}_{j} = i} q_{j} c^{k}_{j}(x_{-j},x^{*}_{j})$ is upper bounded by the right-hand side
of the inequality above. Therefore, Inequality~(\ref{ineq:BCOORD}) follows.
\end{proof}

\begin{theorem}[\cite{Caragiannis:Efficient-coordination-mechanisms}]
The PoA of the game under policy $\textsf{BCOORD}$ is
$O\left(\frac{\log m}{\log\log m}\right)$ by choosing $k = \log m$.
\end{theorem}
\begin{proof}
Let $\vect{y}^{*}$ be an optimal assignment and $\vect{x}^{*}$ be an $m$-efficient 
assignment with property of Lemma~\ref{lem:m-efficient}. Then, for any Nash
equilibrium $\vect{x}$, we have
\begin{align*}
L^{k+1}(\vect{x}) \leq \sum_{i}\sum_{j: x_{j} = i} q_{j} c^{k}_{j}(\vect{x}) 
&\leq O\left(k\alpha^{k}\left(\frac{k}{\log k}\right)^{k-1}\right) \sum_{i}\sum_{j: x^{*}_{j} = i} q_{j} c^{k}_{j}(\vect{x}^{*}) \\
&\leq O\left(k\alpha^{k}\left(\frac{k}{\log k}\right)^{k-1}\right) \cdot m L^{k+1}(\vect{x}^{*}) 
\end{align*}
Therefore, $C(\vect{x}) \leq m^{1/k}L(\vect{x}) \leq O\left(m^{2/k}\cdot \frac{k}{\log k}\right) L(\vect{x}^{*})
\leq O\left(m^{2/k}\cdot \frac{k}{\log k}\right) L(\vect{y}^{*})$.
Choosing $k = \log m$, the theorem follows.
\end{proof}


\subsection{Policy \textsf{CCOORD}, Revisited}

\setcounter{theorem}{8}

\begin{lemma}[\cite{Caragiannis:Efficient-coordination-mechanisms}]	\label{lem:psi}
For any integer $k \geq 1$, any finite set of non-negative reals $A$ and 
any real $b$, the following hold
\begin{enumerate}[(i)]
	\item $L(A)^{k} \leq \Psi_{k}(A) \leq k!L(A)^{k}$
	\item $\Psi_{k}(A)^{k+1} \leq \Psi_{k+1}(A)^{k}$
	\item $\Psi_{k}(A \cup \{b\}) = \Psi_{k}(A) + kb\Psi_{k-1}(A \cup \{b\})$
	\item $\Psi_{k}(A) \leq kL(A)\Psi_{k-1}(A)$
\end{enumerate}
\end{lemma}

In the following, we use also an inequality whose the proof is similar to the one in 
Lemma~\ref{lem:smooth-simple}. For all real positive numbers $a,b$, it holds that
\begin{equation}	\label{ineq:a-b}
ba^{k} \leq \frac{1}{2(k+1)}a^{k+1} + \frac{2^{k}k^{k}}{k+1}b^{k+1}
\end{equation}

\setcounter{theorem}{3}

\begin{lemma}		
Let $\vect{x}$ be a Nash equilibrium. Then, for any $m$-efficient profile $\vect{x}^{*}$, it holds that
$\sum_{i}\sum_{j: x_{j} = i} q_{j}c^{k}_{j}(\vect{x}) 
	\leq  2^{k+1}(k+1)^{k+2} \sum_{i}\sum_{j: x^{*}_{j} = i} q_{j} c^{k}_{j}(\vect{x}^{*})$ 
\end{lemma}
\begin{proof}
We will define recursively two sequences $(\mu_{k})_{k \geq 1}$ and $(\lambda_{k})_{k \geq 1}$
such that $\frac{\lambda_{k}}{1-\mu_{k}} \leq 2^{k+1}(k+1)^{k+2}$ and the following inequalities hold for 
every integer $k$ and every machine $i$
\begin{align}	\label{ineq:ccoord}
\sum_{j: x^{*}_{j} = i} p_{ij} \Psi_{k}(\vect{x}(i) \cup \{j\}) 
\leq \frac{\mu_{k}}{k} \Psi_{k+1}(\vect{x}(i)) 
	+ \frac{\lambda_{k}}{k} \Psi_{k+1}(\vect{x}^{*}(i)) 
\end{align}
If Inequality~(\ref{ineq:ccoord}) holds, then we deduce that
\begin{align*}
\sum_{j: x^{*}_{j} = i}& q_{j}c^{k}_{j}(x_{-j},x^{*}_{j}) 
= \sum_{j: x^{*}_{j} = i} p_{ij} \Psi_{k}(\vect{x}(i) \cup \{j\}) 
\leq \frac{\mu_{k}}{k} \Psi_{k+1}(\vect{x}(i)) 
	+ \frac{\lambda_{k}}{k} \Psi_{k+1}(\vect{x}^{*}(i)) \\
&\leq \mu_{k} L(\vect{x}(i)) \Psi_{k}(\vect{x}(i)) 
	+ \lambda_{k} L(\vect{x}^{*}(i)) \Psi_{k}(\vect{x}^{*}(i)) 
= \mu_{k}\sum_{j: x_{j} = i} q_{j} c^{k}_{j}(\vect{x})
	+ \lambda_{k}\sum_{j: x^{*}_{j} = i} q_{j} c^{k}_{j}(\vect{x}^{*})
\end{align*}
where the second inequality is due to Lemma~\ref{lem:psi}(iv).
Using the smooth argument, the lemma follows. In the following, we recursively define
the sequences $(\mu_{k})_{k \geq 1}$ and $(\lambda_{k})_{k \geq 1}$ with 
the required properties.

For the base case where $k=1$, Inequality~(\ref{ineq:ccoord}) becomes
\begin{align*} \label
\sum_{j: x^{*}_{j} = i} p_{ij} L(\vect{x}(i) \cup \{j\}) 
\leq \mu_{1} \Psi_{2}(\vect{x}(i)) 
	+ \lambda_{1} \Psi_{2}(\vect{x}^{*}(i)) 
\end{align*}
The left-handside is upper-bounded by $L(\vect{x}(i))L(\vect{x}^{*}(i)) + L^{2}(\vect{x}^{*}(i))$. Since
$\left(\frac{L(x(i))}{2} - L(x^*(i)\right)^2>0$, by computation we have $L(\vect{x}(i))L(\vect{x}^{*}(i)) + L^{2}(\vect{x}^{*}(i)) \leq 1/4L^{2}(\vect{x}(i))+ 2L^{2}(\vect{x}^{*}(i))$.  So, by applying  Lemma~\ref{lem:psi}(i)  we obtain

 $$ \sum_{j: x^{*}_{j} = i} p_{ij} \Psi_{k}(\vect{x}(i) \cup \{j\})  \leq 1/4L^{2}(\vect{x}(i))+ 2L^{2}(\vect{x}^{*}(i)) \leq 1/4\Psi_{2}(\vect{x}(i)) + 2\Psi_{2}(\vect{x}^{*}(i))  $$

Choosing $\mu_{1} = 1/4$ and $\lambda_{1} = 2$, Inequality \eqref{ineq:ccoord}  follows.

Suppose the two sequences have been defined until $k-1$. We have
\begin{align*}
\sum_{j: x^{*}_{j} = i} & p_{ij} \Psi_{k}(\vect{x}(i) \cup \{j\}) 
= \sum_{j: x^{*}_{j} = i} p_{ij}\Psi_{k}(\vect{x}(i)) + p_{ij}kp_{ij}\Psi_{k-1}(\vect{x}(i) \cup \{j\}) \\
&\leq L(\vect{x}^{*}(i)) \Psi_{k}(\vect{x}(i)) + kL(\vect{x}^{*}(i)) \sum_{j: x^{*}_{j} = i} p_{ij}\Psi_{k-1}(\vect{x}(i) \cup \{j\}) \\
&\leq L(\vect{x}^{*}(i)) \Psi_{k}(\vect{x}(i)) 
		+ \frac{k}{k-1}L(\vect{x}^{*}(i))\left[\mu_{k-1}\Psi_{k}(\vect{x}(i)) + \lambda_{k-1}\Psi_{k}(\vect{x}^{*}(i)) \right] \\
&\leq \left(1 + \frac{k\mu_{k-1}}{k-1}\right)L(\vect{x}^{*}(i))\Psi_{k}(\vect{x}(i))
		+ \frac{k\lambda_{k-1}}{k-1}L(\vect{x}^{*}(i))\Psi_{k}(\vect{x}^{*}(i)) \\
&\leq \frac{k\mu_{k-1} + k -1}{k-1} \left[\frac{1}{2(k+1)}\Psi_{k}(\vect{x}(i))^{\frac{k+1}{k}} + \frac{2^{k}k^{k}}{k+1}L^{k+1}(\vect{x}^{*}(i))\right]
		+ \frac{k\lambda_{k-1}}{k-1}L(\vect{x}^{*}(i))\Psi_{k}(\vect{x}^{*}(i))	\\
&\leq \frac{k\mu_{k-1} + k -1}{2(k-1)(k+1)} \Psi_{k+1}(\vect{x}(i))
		+ \left(\frac{k\lambda_{k-1}}{k-1} + \frac{(k\mu_{k-1} + k -1)2^{k}k^{k}}{(k+1)(k-1)}\right) \Psi_{k+1}(\vect{x}^{*}(i))		
\end{align*}
where the equality is due to Lemma~\ref{lem:psi}(iii);
the second inequality is due to the induction hypothesis; the fourth one is by applying
inequality~(\ref{ineq:a-b}) for $L(\vect{x}^{*}(i))$ and $\Psi_{k}(\vect{x}(i))^{1/k}$; and the last inequality is 
due to parts (i) and (ii) of Lemma~\ref{lem:psi} (for any set $A$ of non-negative reals, 
$\Psi_{k+1}(A) \geq \Psi_{k}(A) \cdot \Psi_{k}(A)^{1/k} \geq \Psi_{k}(A)L(A)$).

Choosing the sequences $\mu_{k} = k/(k+1)$ and $\lambda_{k} = 2^{k+1}k^{k+1}$ for $k \geq 2$.
The sequences satisfy:
\begin{align*}
	\frac{\mu_{k}}{k} &=  \frac{k\mu_{k-1} + k -1}{2(k-1)(k+1)} \\ 
	\frac{\lambda_{k}}{k} &\geq \frac{k\lambda_{k-1}}{k-1} + \frac{(k\mu_{k-1} + k -1)2^{k}k^{k}}{(k+1)(k-1)} \\
	\frac{\lambda_{k}}{1-\mu_{k}} &\leq 2^{k+1}(k+1)^{k+2}
\end{align*} 
Hence, the lemma follows. 
\end{proof}

\begin{theorem}[\cite{Caragiannis:Efficient-coordination-mechanisms}]
The PoA of the game under policy $\textsf{CCOORD}$ is
$O(\log^{2}m)$ by choosing $k = \log m$.
\end{theorem}
\begin{proof}
Let $\vect{y}^{*}$ be an optimal assignment and $\vect{x}^{*}$ be an $m$-efficient 
assignment with property of Lemma~\ref{lem:m-efficient}. Then, for any Nash
equilibrium $\vect{x}$, we have

\begin{align*}
L(\vect{x}) \max_{i}   \Psi_{k}(\vect{x}(i))   & 
\leq \sum_{i} L(\vect{x}(i))\Psi_{k}(\vect{x}(i))
= \sum_{i}\sum_{j: x_{j} = i} q_{j} c^{k}_{j}(\vect{x}) \\
 &   \leq 2^{k+1}(k+1)^{k+2}\sum_{i}\sum_{j: x^{*}_{j} = i} q_{j} c^{k}_{j}(\vect{x}^{*}) = 2^{k+1}(k+1)^{k+2}\sum_{i} L(\vect{x}^{*}(i))\Psi_{k}(\vect{x}^{*}(i)) \\
 &  
\leq 2^{k+1}(k+1)^{k+2} k! \sum_{i} L(\vect{x}^{*}(i))^{k+1}   \\ &  \leq 2^{k+1}(k+1)^{k+2}k! m L^{k+1}(\vect{x}^{*}) 
\leq 4^{k+1}(k+1)^{k+2}k! m L^{k+1}(\vect{y}^{*}) 
\end{align*}
where the second and the third inequalities are due to Lemma~\ref{lem:ccoord}
and Lemma~\ref{lem:psi}(i). Therefore, 
$\max_{i}\Psi_{k}(\vect{x}(i)) \leq 4^{k+1}(k+1)^{k+2}k! m L^{k+1}(\vect{y}^{*})$
since $L(\vect{x}) \geq L(\vect{y}^{*})$.

Hence, $C(\vect{x}) \leq m^{1/k}\max_{i}\Psi(\vect{x}(i))^{1/k} 
 \leq O\left((m^{2}k!)^{1/k} (k+1)^{(k+2)/k}\right) L(\vect{x}^{*})$.
Choosing $k = \log m - 1$, the theorem follows.
\end{proof}

\end{document}